\documentclass[a4paper,10pt]{article}

\usepackage{lineno,hyperref}

\usepackage[utf8]{inputenc}
\usepackage{graphicx}
\usepackage{algorithm, algorithmicx} 
\usepackage[noend]{algpseudocode}
\usepackage{amssymb}
\usepackage{amsmath}
\usepackage{bm}
\usepackage{amsthm}
\usepackage{geometry}
\usepackage{xcolor}
\usepackage{enumerate}

\theoremstyle{plain}
\newtheorem{definition}{Definition}[section]
\newtheorem{theorem}{Theorem}[section]
\newtheorem{corollary}[theorem]{Corollary}
\newtheorem{lemma}[theorem]{Lemma}

\providecommand{\keywords}[1]{\textbf{\textit{Keywords:}} #1}

\begin{document}
%

\title{Maximum 0-1 Timed Matching on Temporal Graphs\thanks{A preliminary version of this paper \cite{Mandal20} has
appeared in Proceedings of International Conference on Algorithms and Discrete Applied Mathematics (CALDAM 2020).}}
%
%



 \author{Subhrangsu Mandal \thanks{{\em Department of Computer Science and Engineering, Indian Institute of Technology Kharagpur, Kharagpur 721302, West Bengal, India, subhrangsum@cse.iitkgp.ac.in}.} \and 
 Arobinda Gupta \thanks{{\em Department of Computer Science and Engineering, Indian Institute of Technology Kharagpur, Kharagpur 721302, West Bengal, India, agupta@cse.iitkgp.ac.in}.}}

\maketitle

\begin{abstract}
Temporal graphs are graphs where the topology and/or other properties of the graph change with time. They have been used to model applications with temporal information in various domains. Problems on static graphs become more challenging to solve in temporal graphs because of dynamically changing topology, and many recent works have explored graph problems on temporal graphs.  In this paper, we define a type of matching called {\em 0-1 timed matching} for temporal graphs, and investigate the problem of finding a {\em maximum 0-1 timed matching} for different classes of temporal graphs. We assume that only the edge set  of the temporal graph changes with time. Thus, a temporal graph can be represented by associating each edge with one or more non-overlapping discrete time intervals for which that edge exists. We first prove that the problem is NP-complete for rooted temporal trees when each edge is associated with two or more time intervals. We then propose an $O(n \log n)$ time algorithm for the problem on a rooted temporal tree with $n$ vertices when each edge is associated with exactly one time interval. The problem is then shown to be NP-complete also for bipartite temporal graphs even when each edge is associated with a single time interval and degree of each vertex is bounded by a constant $k \geq 3$. We next investigate approximation algorithms for the problem for temporal graphs where each edge is associated with more than one time intervals. It is first shown that there is no $\frac{1}{n^{1-\epsilon}}$-factor approximation algorithm for the problem for any $\epsilon > 0$ even on a rooted temporal tree with $n$ vertices unless NP = ZPP. We then present a $\frac{5}{2\mathcal{N}^* + 3}$-factor approximation algorithm for the problem for general temporal graphs where $\mathcal{N^*}$ is the average number of edges overlapping in time with each edge in the temporal graph. The same algorithm is also a constant-factor approximation algorithm for temporal graphs with degree of each vertex bounded by a constant.

\noindent \keywords{ 0-1 timed matching, temporal matching, time dependent matching, temporal graph, time varying graph}
\end{abstract}


\section{Introduction}

Graphs are an important tool for modelling systems with a set of objects and pairwise relationships between those objects. In many applications, the objects and the relations between them have temporal properties, resulting in dynamically changing graphs where the topology and/or other properties of the graph change with time. Static graphs are not suitable for modelling such applications as they cannot represent the temporal information. Some examples of domains where such dynamic graphs with temporal properties arise are  networks with intermittent inter-vertex connectivity such as delay tolerant networks (DTN) \cite{xie16}, vehicular ad-hoc networks (VANET) \cite{Feng17}, mobile ad-hoc networks (MANET) \cite{Ferreira10}, social networks \cite{iribarren09, Amblard11, bravo19}, biological networks \cite{lebre10,Han04,rao07}, and transportation networks \cite{lordan20, halpern74}. While modelling these systems, the temporal information in the system must be represented in the graph model used. Temporal graphs \cite{Kostakos09} have been proposed as a tool for modelling such systems with temporal properties.

Solving many well known graph problems on temporal graphs introduces new challenges. For many graph problems, the usual definitions and algorithms for the problem on static graphs do not apply 
directly to temporal graphs. As an example, consider the notion of paths in a graph, which is used extensively in applications using graphs. Unlike static graphs, in a temporal graph, an edge can only be used for traversal if it exists at the time of the traversal, and the definition of a path in a static graph no longer holds in a temporal graph. In a temporal graph, a path from a vertex $u$ to vertex $v$ must contain a sequence of edges from $u$ to $v$ that exist in strictly increasing order of time, unlike a static graph where just a sequence of edges from $u$ to $v$ is sufficient. Various other notions of paths  have been defined on temporal graphs \cite{CasteigtsHMZ20, Xuan03, Wu14}. This makes many problems that use notions of paths  more complex to solve in temporal graphs. Similar examples exist for many other graph related problems. Several works have addressed graph problems such as finding paths and trees \cite{Xuan03, Huang15, MandalG20cctree}, computing dominating sets \cite{Mandal18}, travelling salesman problem \cite{Michail16}, finding vertex separator \cite{zschoche20}, computing sparse spanners \cite{casteigts19}, etc. on temporal graphs. 

In this paper, we investigate the problem of finding matchings in a temporal graph. A matching in a static graph is defined as a subset of edges of the graph such that no two edges share a common vertex. A maximum matching is a matching with the maximum cardinality among all matchings. Finding a maximum matching for a static graph \cite{Hopcroft71, Micali80} is a well-studied problem in the area of static graphs due to its wide application. Assignment problems \cite{bertsekas81} form an important class of applications that use the notion of maximum matching. Some example applications of assignment problems that occur in dynamic graph topologies are task assignment problem in ad allocation \cite{bhalgat12, zhang18}, crowdsourcing market \cite{Ho12}, dynamic assignment problem \cite{spivey04} etc. Some of these applications may require that assignments be available for all applicable time periods. As an example, consider a problem of matching the maximum number of consumers with service providers, where each consumer requests for service in a number of time slots, and may want to accept a matched service provider only if it can get the service at all its requested time slots. The traditional definition of matching for static graphs is not enough to address such assignment problems due to its temporal nature. However, this can be modelled by a temporal graph with edges between consumers and service providers, with each edge marked with the time intervals for which the service is requested. An allocation is valid if it ensures that no two consumers allocated to a service provider has any common requested time slot. In this paper, we define a type of matching called  \textit{0-1 timed matching} for a temporal graph to model such situations, and investigate the complexity and algorithms for finding a \textit{maximum 0-1 timed matching} for different types of temporal graphs. We have assumed that only the edge set of the temporal graph changes with time. Thus, a temporal graph can be represented by labelling each edge with non-overlapping discrete time intervals for which that edge exists. The underlying graph of a temporal graph is a static graph for which the vertex set includes all the vertices of the temporal graph and the edge set includes each edge which is present in the temporal graph for at least one timestep. The specific contributions of this paper are as follows.
\begin{enumerate}
 \item We prove that the problem of finding a maximum 0-1 timed matching for a rooted temporal tree is NP-complete when each edge of the tree is associated with $2$ or more time intervals.
 \item We show that the problem is solvable in polynomial time if each edge of the rooted temporal tree is associated with a single time interval. In particular, we propose a dynamic programming based $O(n \log n)$ time algorithm to solve the problem on such a rooted temporal tree with $n$ vertices.
 \item Next, we study the computational complexity of the problem when each edge of the temporal graph is associated with a single time interval. We prove that the problem is NP-complete in this case even for bounded degree bipartite temporal graphs when degree of each vertex is bounded by $3$ or more. This automatically proves that the problem is NP-complete for a bipartite temporal graph with a single time interval per edge, and hence, for a general temporal graph with a single time interval per edge.  
 \item We investigate the hardness of approximation of the problem when each edge of the temporal graph is associated with multiple time intervals. We prove that there is no approximation algorithm with approximation ratio $\frac{1}{n^{1-\epsilon}}$, for any $\epsilon > 0$, for finding a maximum 0-1 timed matching even on a rooted temporal tree with $n$ vertices when each edge is associated with multiple time intervals unless NP = ZPP.
 \item We propose an approximation algorithm to address the problem for a temporal graph when each edge is associated with multiple time intervals. The approximation ratio of the proposed algorithm is $\frac{5}{2 \mathcal{N}^* + 3}$ where $\mathcal{N^*}$ is the average number of edges overlapping with each edge in the temporal graph. Two edges are overlapping with each other if both are incident on the same vertex and there exists at least one timestep when both the edges exist\footnote{Formal definition of overlapping edges is given in Section \ref{defin}.}. We also show that the same algorithm is a constant factor approximation algorithm for a temporal graph when degree of each vertex is bounded by a constant. 
 \end{enumerate}

The rest of this paper is organised as follows. Section \ref{relWork} describes some related work in the area. Section \ref{model} describes the system model used. Section \ref{defin} formally defines the problem. Section \ref{matchingTree} presents the results related to the problem of finding a maximum 0-1 timed matching for a rooted temporal tree. Section \ref{bipartite} presents the results related to the problem of finding a maximum 0-1 timed matching for a general temporal graph. Finally Section \ref{conclusion} concludes the paper.

\section{Related Work}
\label{relWork}

The problem of finding a maximum matching is a well studied problem for static graphs. In \cite{Edmonds65}, Edmonds proposed a $O(n^4)$ time algorithm, where $n$ is the number of vertices in the input graph. Since then, many algorithms have been proposed to address the problem on both general graphs and other restricted classes of graphs \cite{Hopcroft71,Micali80,EvenT75,Even75,Kameda74,Cheriyan97,Mucha04,Mucha06,Mulmuley87}. One generalisation of the maximum matching problem is the maximum packing problem where the goal is to find the maximum number of vertex disjoint subgraphs of a given graph such that each subgraph is isomorphic to an element of a given class of graphs. This is a well studied problem for static graphs, with many different variants of this problem addressed for both general graphs and other restricted classes of graphs \cite{kierstead09, loebl90, moser09, garcia02}.


Matchings in temporal graphs have recently attracted the attention of researchers. One variant of the matching problem on temporal graphs is {\em multistage matching} \cite{chimani20}. In this problem, the temporal graph is viewed as a sequence of static graphs, one for each timestep of the total duration, with the static graph for timestep $t$ containing all the vertices and edges that exist at $t$. A sequence of maximum matchings is then found, one for each static graph in this representation of the temporal graph. The work in  \cite{chimani20} attempts to optimize different parameters such as {\em intersection profit} and {\em union cost} while finding this sequence of matchings. In the case of intersection profit, the goal is to minimize the number of changes incurred to construct a maximum matching for a static graph at a certain timestep from a maximum matching for the static graph at the previous timestep. In the case of union cost, the goal is to minimize the total number of edges in the union of the maximum matchings for the graphs at each timestep. Several other works \cite{gupta14, bampis18} have studied different variants of multistage matchings.   

Other than multistage matching, a few other definitions of matchings for temporal graphs are available in literature. For a given temporal graph, Michail et al. \cite{Michail16} consider the decision problem of finding if there exists a maximum matching $M$ in the underlying graph such that a single label can be assigned to each edge of $M$ with the constraint that the assigned label to an edge is chosen from the timesteps when that edge exists in the temporal graph and no two edges of $M$ are assigned the same label. It is then proved that the problem is NP-hard. Baste et al. \cite{Baste19}  define another type of temporal matching called $\gamma$-matching. $\gamma$-edges are defined as edges which exist for at least $\gamma$ consecutive timesteps. The maximum $\gamma$-matching is defined as a maximum cardinality subset of $\gamma$-edges such that no two $\gamma$-edges in the subset share any vertex at any timestep. They have shown that the problem of finding a maximum $\gamma$-matching is NP-hard when $\gamma > 1$, and proposed a 2-approximation algorithm for the problem. Mertzios et al. \cite{mertzios20} define another type of temporal matching called $\Delta$-matching where two edge instances at timesteps $t, t'$ can be included in a matching if either those two edge instances do not share any vertex or $|t - t'|$ is greater than or equal to a positive integer $\Delta$. They prove that this problem is APX-hard for any $\Delta \geq 2$ when the lifetime of the temporal graph is at least $3$, and the problem remains NP-hard even when the underlying graph of the temporal graph is a path. An approximation algorithm is proposed to find a $\frac{\Delta}{2\Delta - 1}$-approximate maximum $\Delta$-matching for a given temporal graph with $n$ vertices, $m$ edges and lifetime $\mathcal{T}$ in $O(\mathcal{T}m(\sqrt{n}+ \Delta))$ time. 
In \cite{akrida20}, Akrida et al. address the maximum matching problem on stochastically evolving graphs represented using stochastic arrival departure model. In this model, each vertex in the temporal graph arrives and departs at certain times, and these arrival and departure times of each vertex are determined using independent probability distributions. The vertex exists in the time interval between the arrival time and the departure time. An edge between two vertices can exist if there is an intersection between the time intervals for which those two vertices exist. A matching on a stochastically evolving graph is defined as the subset of edges such that no two edges are incident on the same vertex. A fully randomized approximation scheme (FPRAS) has been derived to approximate the expected size of maximum cardinality matching on a stochastically evolving graph. A probabilistic optimal algorithm is proposed when the model is defined over two timesteps. They have also defined price of stochasticity and proved that the upper bound on the price of stochasticity is $\frac{2}{3}$. 

In this paper, we propose another type of matching for temporal graphs called {\em 0-1 timed matching}. We then investigate the complexity and algorithms for finding a maximum 0-1 timed matching in temporal graphs. 

\section{System Model}
\label{model}
We represent a temporal graph by the {\em evolving graphs} \cite{Ferreira02} model. In this model, a temporal graph is represented as a finite sequence of static graphs, each static graph being an undirected graph representing the graph at a discrete timestep. The total number of timesteps is called the {\em lifetime} of the temporal graph. In this paper, we assume that the vertex set of the temporal graph remains unchanged throughout the lifetime of the temporal graph; only the edge set changes with time. For simplicity, it is assumed that all the changes in the edge set are known a-priori (this assumption has also been used in other existing works \cite{Xuan03, mertzios20}). Also, there are no self-loops and at most one edge exists between any two vertices at any timestep. Thus, a temporal graph $\mathcal{G = (V, E)}$ with set of vertices $\mathcal{V}$ and set of edges $\mathcal{E}$ is represented as a sequence of graphs $(G_0, G_1,$ $\cdots, G_{\mathcal{T} - 1})$, where $G_i = (\mathcal{V}, \mathcal{E}_i)$ is the static graph at timestep $i$ with set of edges $\mathcal{E}_i$ that exist at timestep $i$, and $\mathcal{T}$ is the lifetime of $\mathcal{G}$. As only the edge set changes with time, each edge in $\mathcal{E}$ of a temporal graph $\mathcal{G}$ can be represented by specifying the time intervals for which the edge exists. Thus an edge $e \in \mathcal{E}$ between vertices $u$ and $v$ can be represented as $e(u, v, (s_1, f_1),$ $(s_2, f_2), \cdots, (s_k, f_k))$, where $u, v$ $\in$ $\mathcal{V}, u \neq v$, $f_k \leq \mathcal{T}$ and a pair $(s_i, f_i)$ indicates that the edge exists for the time interval $[s_i, f_i)$, where $0 \leq s_i < f_i \leq \mathcal{T}$. In other words, each interval $(s_i, f_i)$ associated with edge $e$ between vertices $u$ and $v$ denotes that the edge $e$ exists in all the static graphs $(G_{s_i}, G_{s_i+1}, \cdots, G_{f_i - 1})$. Also, if an edge $e$ has two such pairs $(s_i, f_i)$ and $(s_j, f_j)$, $s_i \neq s_j$ and if $s_i$ $<$ $s_j$ then $f_i$ $<$ $s_j$. Thus, the maximum number of time intervals for an edge can be $\lfloor \frac{\mathcal{T}}{2} \rfloor$. An edge at a single timestep is called an instance of that edge. For simplicity, we also denote the edge $e$ between vertices $u, v \in \mathcal{V}$ by $e_{uv}$ when the exact time intervals for which $e$ exists are not important. The corresponding instance of $e_{uv}$ at time $t$ is denoted by $e^t_{uv}$. 

\section{Problem Definition}
\label{defin}
In this section, we define a 0-1 timed matching for temporal graphs. We first define some terminologies related to temporal graphs that we will need. 

For a given temporal graph $\mathcal{G = (V, E)}$ with lifetime $\mathcal{T}$, the {\em underlying graph} of $\mathcal{G}$ is defined as the static graph $\mathcal{G}_U = (\mathcal{V}, \mathcal{E}_U)$, where $\mathcal{E}_U =  \{(u, v)\,|\, \exists t$ such that, $e^t_{uv}$ is an instance of $e_{uv} \in \mathcal{E} \}$. Next, we define different types of temporal graphs. 

\begin{definition}
 \textbf{Temporal Tree:} A temporal graph $\mathcal{G}$ is a temporal tree if the underlying graph of $\mathcal{G}$ is a tree. 
\end{definition}

\begin{definition}
 \textbf{Rooted Temporal Tree:} A rooted temporal tree $\mathcal{G = (V, E)}$ rooted at vertex $r$ is a temporal tree with one vertex $r \in \mathcal{V}$ chosen as root of the tree.
\end{definition}

Note that, the underlying graph of a rooted temporal tree is also a rooted tree. For any vertex $v$, let $p(v)$ denote the parent vertex of $v$ and $child(v)$ denote the set of children vertices of $v$ in the underlying graph of the rooted temporal tree. For the root vertex $r$, $p(r) = NULL$. {\em Depth} of a vertex $v$ in a rooted temporal tree $\mathcal{G}$ rooted at $r$ is the path length from $r$ to $v$ in $\mathcal{G}_U$. {\em Height} of a rooted temporal tree $\mathcal{G}$ is the maximum depth of any vertex in $\mathcal{G}$.

 \begin{figure}
 \begin{center}
  \includegraphics[scale=.50]{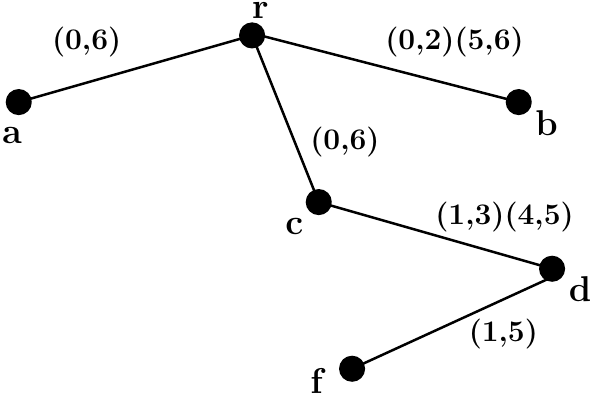}
  \caption{A temporal tree rooted at vertex $r$}
  \label{fig:tree}
 \end{center} 
 \end{figure}

Figure \ref{fig:tree}, shows a temporal tree rooted at vertex $r$. In this temporal tree $p(a), p(b), p(c)$ is $r$ and $child(r) = \{a, b, c\}$. The depth of vertex $r$ is $0$. Depth of $a, b, c$ is $1$, depth of $d$ is $2$ and depth of $f$ is $3$. The height of this rooted temporal tree is $3$.

\begin{definition}
 \textbf{Bipartite Temporal Graph:} A temporal graph $\mathcal{G}$ is bipartite if the underlying graph of $\mathcal{G}$ is a bipartite graph.
\end{definition}

\begin{definition}
 \textbf{Bounded Degree Temporal Graph:} A temporal graph $\mathcal{G}$ is a bounded degree temporal graph where degree of each vertex is bounded by some positive integer $k$, if the underlying graph of $\mathcal{G}$ is a bounded degree graph where degree of each vertex is bounded by $k$.
\end{definition}

\begin{definition}
 \textbf{Bounded Degree Bipartite Temporal Graph:} A temporal graph $\mathcal{G}$ is a bounded degree bipartite temporal graph if it is both a bipartite temporal graph and a bounded degree temporal graph.
\end{definition}

\begin{definition}
 \textbf{Overlapping Edge:} Given a temporal graph $\mathcal{G = (V, E)}$, any edge $e_{vw} \in \mathcal{E}$ is said to be overlapping with another edge $e_{uv}$ if there exists a timestep $t$ such that both $e^t_{vw}$ and $e^t_{uv}$ exist.  
\end{definition}

Note that if $e_{uv}$ is overlapping with edge $e_{vw}$, then $e_{vw}$ is also overlapping with $e_{uv}$. We refer to such pair of edges as {\em edges overlapping with each other}. 

In Figure \ref{fig:bipartite}, $e_{dg}$ is an overlapping edge with $e_{fg}$ because both edges are incident on $g$ and both $e^1_{dg}$ and $e^1_{fg}$ exist. On the other hand, edges $e_{ab}$ and $e_{ad}$ are incident on the same vertex $a$, but there is no timestep $t$ when both $e^t_{ab}$ and $e^t_{ad}$ exist. Thus $e_{ab}$ and $e_{ad}$ are {\em non-overlapping with each other}. For any two sets of edges $E_1$, $E_2$, if $E_1 \subseteq E_2$ and no two edges in $E_1$ are overlapping with each other, then $E_1$ is called a {\em non-overlapping subset} of $E_2$.

\begin{definition}
	\textbf{Overlapping Number of Edge $e_{uv}$:} Given a temporal graph $\mathcal{G = (V, E)}$, the overlapping number of an edge $e_{uv} \in \mathcal{E}$, denoted by $\mathcal{N}(e_{uv})$, is the number of edges overlapping with $e_{uv}$ in $\mathcal{E}$. 
\end{definition}

\begin{definition}
 \textbf{0-1 Timed Matching:} A 0-1 timed matching $M$ for a given temporal graph $\mathcal{G = (V, E)}$ is a non-overlapping subset of $\mathcal{E}$. 
\end{definition}

\begin{definition}
 \textbf{Maximum 0-1 Timed Matching:} A maximum 0-1 timed matching for a given temporal graph is a 0-1 timed matching with the maximum cardinality.
\end{definition}

\begin{definition}
 \textbf{Maximal 0-1 Timed Matching:} A 0-1 timed matching $M$ for a given temporal graph $\mathcal{G = (V, E)}$ is maximal if for any edge $e_{uv} \in \mathcal{E} \setminus M$, $M \cup \{e_{uv}\}$ is not a 0-1 timed matching for $\mathcal{G}$.
\end{definition}

 \begin{figure}
 \begin{center}
  \includegraphics[scale=.70]{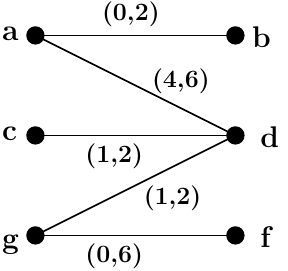}
  \caption{A bipartite temporal graph}
  \label{fig:bipartite}
 \end{center} 
 \end{figure}
 For the bipartite temporal graph $\mathcal{G}$ shown in Figure \ref{fig:bipartite}, a maximum 0-1 timed matching $M$ is $\{e_{ab}, e_{ad}, e_{cd}, e_{fg}\}$. $M$ is also a maximal 0-1 timed matching for $\mathcal{G}$. Consider another 0-1 timed matching $M' = \{e_{ab}, e_{ad}, e_{dg}\}$ for $\mathcal{G}$. $M'$ is a maximal 0-1 timed matching for $\mathcal{G}$ but not a maximum 0-1 timed matching. Note that, the edges in a 0-1 timed matching for a given temporal graph may not be a matching for its underlying graph. 

 In the next section, we explore the problem of finding a maximum 0-1 timed matching for a given rooted temporal tree.

\section{Finding a Maximum 0-1 Timed Matching for Rooted Temporal Tree} 
\label{matchingTree}

In this section, we first analyse the hardness of computing a maximum 0-1 timed matching for a given rooted temporal tree $\mathcal{G = (V, E)}$. We show that this problem is NP-complete for $\mathcal{G}$ when each edge in $\mathcal{E}$ is associated with $2$ or more time intervals. We then explore the problem when each edge of $\mathcal{E}$ is associated with a single time interval. We find that this problem is solvable in polynomial time and propose a dynamic programming based algorithm for it.

\subsection{Complexity of Finding a Maximum 0-1 Timed Matching for Rooted Temporal Tree}

We first show that the problem of finding a maximum {\em 0-1 timed matching} is NP-complete for a rooted temporal tree even 
when the number of intervals associated with each edge is at most $2$. We refer to the problem of finding a maximum 0-1 timed 
matching for a rooted temporal tree when each edge is associated with at most $2$ time intervals as {\em MAX-0-1-TMT-2}. We first 
define the decision version of MAX-0-1-TMT-2, referred to as the {\em D-MAX-0-1-TMT-2} problem.

\begin{definition}
 \textbf{D-MAX-0-1-TMT-2:} Given a rooted temporal tree $\mathcal{G = (V, E)}$ with lifetime $\mathcal{T}$, where each edge in $\mathcal{E}$ is associated with at most $2$ time intervals, and a positive integer $g$, does there exist a 0-1 timed matching $M$ for $\mathcal{G}$ such that $|M| = g$?
\end{definition}

We show that there is a polynomial time reduction from the decision version of the problem of finding a {\em maximum rainbow matching for a properly edge coloured path}, referred to as D-MAX-RBM-P \cite{Le14}, to the D-MAX-0-1-TMT-2 problem. Before describing the details of the reduction, we define a {\em properly edge coloured path} and the D-MAX-RBM-P problem. 

\begin{definition}
	\textbf{Properly Edge Coloured Path:} A path $P$ is a properly edge coloured path if each edge of $P$ is coloured in such a way that no two edges incident on the same vertex are coloured with the same colour.
\end{definition}	

\begin{definition}
 \textbf{D-MAX-RBM-P:} Given a properly edge coloured path $P = (V, E)$ and a positive integer $h$, does there exist a set $R \subset E$ of size $h$, such that $R$ is a matching for $P$ and no two edges in $R$ are coloured with the same colour?
\end{definition}

\noindent The D-MAX-RBM-P problem is known to be NP-complete \cite{Le14}.

\begin{theorem}
 D-MAX-0-1-TMT-2 is NP-complete.
 \label{thm:NPTree}
\end{theorem}

\begin{proof}
 We first show that the problem is in NP. Consider a certificate $\langle \langle \mathcal{G = (V,E)},$ $g \rangle, M\rangle$, where $\mathcal{G}$ is a rooted temporal tree with lifetime $\mathcal{T}$, each edge in $\mathcal{E}$ is associated with at most 2 time intervals, $g$ is a given integer and $M$ is a subset of $\mathcal{E}$. We consider one edge $e_{uv} \in M$ at a time and compare associated time intervals of $e_{uv}$ with associated time intervals of all the other edges in $M$ to find any edge overlapping with $e_{uv}$ in $M$. We perform this check for all the edges in $M$ to find any such overlapping edges with each other. This checking can be done in polynomial time. Whether $|M| = g$ can also be easily checked in polynomial time. Hence, the D-MAX-0-1-TMT-2 problem is in NP.

 Next, we prove that there is a polynomial time reduction from the D-MAX-RBM-P problem to the D-MAX-0-1-TMT-2 problem. Consider an instance $\langle P = (V, E), h \rangle$ of the D-MAX-RBM-P problem where $P$ is a properly edge coloured path, $V = \{v_0, v_1,\cdots,v_{n-1}\}$, $|V| = n$, $E = \{(v_0, v_1), (v_1, v_2), \cdots,$ $(v_{n-2}, v_{n-1})\}$, $|E| = n-1$, and $h$ is a positive integer. For our reduction, we assign a sequence of positive integers $A = \{a_i\,|\,a_i \in \mathbb{N}, a_i < n\}$, to vertices in $V$ in the following way. We assign $a_i = i$ to the vertex $v_i$. The path $P$ is a properly edge coloured path. Let $c$ be the number of different colours used to colour $P$. As the number of edges in $P$ is $n-1$, then $c \leq n-1$. We represent these colours with different integers from $n$ to $n+c-1$.

 From this given instance of the D-MAX-RBM-P problem, we construct an instance of the D-MAX-0-1-TMT-2 problem as follows. 
 \begin{itemize}
    \item The temporal graph $\mathcal{G = (V, E)}$ is constructed as follows.
    \begin{itemize}
        \item We add a vertex $\nu(v_iv_{i+1})$ for each edge $(v_i, v_{i+1}) \in E$. Additionally we add a vertex $r$. Thus, 
        
        $\mathcal{V}$ $:=$ $\{\nu(v_iv_{i+1})\,|\,\forall (v_i, v_{i+1}) \in E\} \cup \{r\}$.
        \item We add an edge between each added vertex $\nu(v_iv_{i+1}) \in \mathcal{V}$ and $r$. This edge exists for time intervals $(a_i, a_i+2)$ and $(c_j, c_j+1)$, where $a_i$ is the integer assigned to $v_i$, $c_j$ is the integer representing the colour by which $(v_i, v_{i+1})$ is coloured. Thus,
        
        $\mathcal{E}$ $:=$ $\{e(\nu(v_iv_{i+1}), r, (a_i, a_i+2), (c_j, c_j+1))\,|\, \forall (v_i, v_{i+1}) \in E\}$,
        \item $\mathcal{T}$ $:=$ $n+c$
    \end{itemize}
    \item $g := h$
 \end{itemize}
 \begin{figure}
 \begin{center}
  \includegraphics[scale=.50]{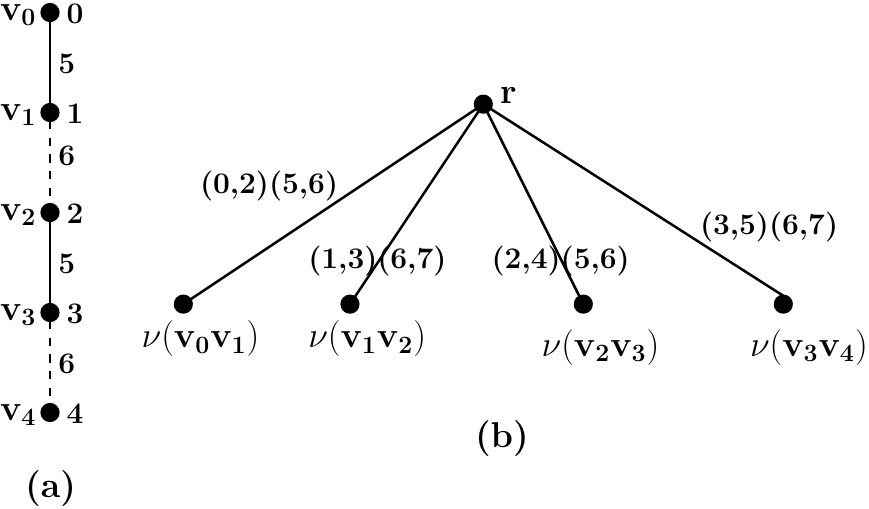}
  \caption{(a) A properly edge coloured path where vertices are assigned some integers (dashed and continuous lines are representing two different colours; integer along an edge is representing the colour of it), (b) Corresponding temporal tree rooted at vertex $r$.}
  \label{RBMatching}
 \end{center} 
 \end{figure}
 
 As each edge in $\mathcal{E}$ connects a vertex in $\mathcal{V} \setminus \{r\}$ to $r$, $\mathcal{G}$ is a temporal tree rooted at $r$ (we choose $r$ as the root vertex). According to the construction, the number of intervals associated with each edge in $\mathcal{E}$ is $2$. Any edge $e(u, v, (s_1, f_1) (s_2, f_2)) \in \mathcal{E}$ is also denoted as $e_{uv}$ when the time intervals for which this edge exists are not important. Figure \ref{RBMatching} shows the construction of a rooted temporal tree from a properly edge coloured path.  

 We first show that, if there is a solution for the instance of the D-MAX-0-1-TMT-2 problem, then there is a solution for the instance of the D-MAX-RBM-P problem. For a 0-1 timed matching $M$, $|M| = g$, for $\mathcal{G}$, we  construct a rainbow matching $R$, $|R| = h$, for $P$ as follows. Consider the set of edges $R = \{(v_i, v_{i+1})\,|\,e_{r\nu(v_iv_{i+1})} \in M \}$. As $|M| = g$ and $g = h$, $|R| = h$. We prove that $R$ is a rainbow matching for $P$. We prove this by contradiction. Assume that, $R$ is not a rainbow matching for $P$. This is possible in two cases.
 \begin{enumerate}[I.]
     \item There are at least two edges $(v_{i-1}, v_i), (v_i, v_{i+1})$ incident on the same vertex $v_i$ included in $R$. As $P$ is a path, according to the construction, $v_{i-1}, v_i, v_{i+1}$ are assigned three consecutive integers $a_{i-1}, a_i, a_{i+1}$ respectively. This implies that, time intervals $(a_{i-1}, a_{i-1}+2)$ and $(a_i, a_i+2)$ are associated with edges $e_{r\nu(v_{i-1}v_i)}$ and $e_{r\nu(v_iv_{i+1})}$ respectively, and both are included in $M$. As $a_i = a_{i-1}+1$, edges $e_{r\nu(v_{i-1}v_i)}$ and $e_{r\nu(v_iv_{i+1})}$ exist at the timestep $a_i$. This contradicts the fact that $M$ is a 0-1 timed matching for $\mathcal{G}$.
     \item There are two edges $(v_i, v_{i+1}), (v_j, v_{j+1})$ which are coloured with the same colour, $c_k$. This implies that edges $e_{r\nu(v_iv_{i+1})}$ and $e_{r\nu(v_jv_{j+1})}$ are included in $M$ when both exist at timestep $c_k$. This also contradicts the fact that $M$ is a 0-1 timed matching for $\mathcal{G}$. 
 \end{enumerate}
 
 Next, we show that if there is no solution for the instance of the D-MAX-0-1-TMT-2 problem, then there is no solution for the instance of the D-MAX-RBM-P problem. To show this, we prove that if we have a solution for the instance of the D-MAX-RBM-P problem, then we have a solution for the instance of the D-MAX-0-1-TMT-2 problem. For a rainbow matching $R$, $|R| = h$, for $P$, we construct a 0-1 timed matching $M$, for $\mathcal{G}$ as follows. Consider the set $M = \{e_{r\nu(v_iv_{i+1})} \,|\, (v_i, v_{i+1}) \in R\}$. As $|R| = h$ and $h = g$, $|M| = g$. Next we show that, $M$ is a 0-1 timed matching for $\mathcal{G}$. We prove this by contradiction. Assume that, $M$ is not a 0-1 timed matching for $\mathcal{G}$. This implies that, there are at least two edges $e_{r\nu(v_iv_{i+1})}$ and $e_{r\nu(v_jv_{j+1})}$ in $M$ such that there exists at least one timestep $t$ when both the edges exist. According to our construction of $\mathcal{G}$, $t < n+c$. There can be two possible cases. 
 \begin{enumerate}[I.]
     \item {\em $0 \leq t \leq n-1$:} As $0 \leq t \leq n-1$, $t$ is assigned to some vertex in $V$. According to our assumption, edges $e_{r\nu(v_iv_{i+1})}$ and $e_{r\nu(v_jv_{j+1})}$ in $M$ and both exist at timestep $t$. This implies that, timestep $t$ is present in both intervals $(a_i, a_i+2)$ and $(a_j, a_j+2)$ which are associated with $e_{r\nu(v_iv_{i+1})}$ and $e_{r\nu(v_jv_{j+1})}$ respectively (according to construction of $\mathcal{G}$). Without loss of generality, we can assume that $a_i < a_j$. In this scenario, $a_i+1 = a_j = t$. Then, according to the assignment of integers to the vertices in $V$, $v_{i+1} = v_j$. This implies that, the edges $(v_i, v_{i+1}), (v_j, v_{j+1}) \in R$ are incident on the same vertex $v_{i+1}$. This contradicts the fact that $R$ is a rainbow matching for $P$.  
     \item {\em $n \leq t \leq n+c-1$:} As $n \leq t \leq n+c-1$, $t$ represents one colour which is used to colour the edges in $E$. According to the assignment of integers to the colours of each edge, $n \leq t \leq n+c-1$ implies that edges $(v_i, v_{i+1}), (v_j, v_{j+1}) \in R$ are coloured with the same colour which is represented by $t$. This also contradicts the fact that $R$ is a rainbow matching for $P$.
 \end{enumerate}
 This completes the proof of this theorem.
\end{proof}

\subsection{Finding a Maximum 0-1 Timed Matching for a Rooted Temporal Tree with Single Time Interval per Edge}
We present a dynamic programming based algorithm for finding a maximum 0-1 timed matching for a rooted temporal tree $\mathcal{G = (V,E)}$ with root $r \in \mathcal{V}$ where each edge in $\mathcal{E}$ is associated with a single time interval. 

Consider a vertex $v \in \mathcal{V}$ with parent $p(v)$. Let $\mathbb{T}_v = (\mathbb{V}_v, \mathbb{E}_v)$ be the temporal subtree rooted at $v$. Consider the temporal tree  $T_v = (\mathcal{V}_v, \mathcal{E}_v)$, where $\mathcal{V}_v = \mathbb{V}_v \cup \{p(v)\}$, $\mathcal{E}_v = \mathbb{E}_v \cup \{e_{vp(v)}\}$. For the purpose of computing the depths of vertices in $T_v$, we consider $p(v)$ as the root of $T_v$. Let $M_v$ denote a maximum 0-1 timed matching for $T_v$.  For any non-root vertex $v$,  two cases are possible for $M_v$.

\begin{enumerate}
	\item Edge $e_{vp(v)}$ is not included in $M_v$.
	\item Edge $e_{vp(v)}$ is included in $M_v$.
\end{enumerate}

Let $TM1[v]$ and $TM2[v]$ denote the maximum 0-1 timed matchings for $T_v$ which does not include edge $e_{vp(v)}$  (Case 1), and which includes edge $e_{vp(v)}$ (Case 2) respectively. In the rest of this paper, $T_{u_i}=(\mathcal{V}_{u_i}, \mathcal{E}_{u_i})$ denotes the temporal tree with vertex set $\mathcal{V}_{u_i} = \mathbb{V}_{u_i} \cup \{p(u_i)\}$, and edge set $\mathcal{E}_{u_i} = \mathbb{E}_{u_i} \cup \{e_{u_ip(u_i)}\}$ where $\mathbb{T}_{u_i} = (\mathbb{V}_{u_i}, \mathbb{E}_{u_i})$ is the temporal subtree rooted at $u_i$. The algorithm first orders the vertices in $\mathcal{V}$ in non-increasing order of depths. It then computes $TM1[v]$ and $TM2[v]$ for each $T_v$, $v \in \mathcal{V}$ in this order.

We first describe the method for computing $TM1[v]$ for $T_v$ when for all $T_{u_i}$, $u_i \in child(v)$, $TM1[u_i]$ and $TM2[u_i]$ are already computed. Note that for any $T_{u_i}$, $u_i \in child(v)$, if $|TM2[u_i]| < |TM1[u_i]| + 1$, then there is a maximum 0-1 timed matching for $T_v$ which does not include $e_{vu_i}$\footnote{This statement is formally proved in Lemma \ref{lem:lessTM2} later in the paper.}. Thus, for any $u_i \in child(v)$, the edge $e_{vu_i}$ can be considered for inclusion in $M_v$ only if $|TM2[u_i]| = |TM1[u_i]| + 1$. Based on this, we define the set of {\em feasible edges} as follows.

\begin{definition}
    \label{def:feasible}
	\textbf{Feasible Edges for $T_v$:} The feasible edges for $T_v$, denoted by $F_v$, is the set of edges $e_{vu_i}$ such that $u_i \in child(v)$ and $|TM2[u_i]| = |TM1[u_i]| + 1$.  
\end{definition} 

To compute $TM1[v]$ for $T_v$, the algorithm first computes the feasible edges $F_v$ for $T_v$. The algorithm then finds the maximum cardinality subset $F'_v \subseteq F_v$ such that no two edges in $F'_v$ are overlapping with each other.  $TM1[v]$ is then computed using the following equation.

\begin{equation}\label{eqn1}
	TM1[v] := (\bigcup_{\forall u_i \in child(v) \land e_{vu_i} \in F'_v} TM2[u_i]) \;\cup\; (\bigcup_{\forall u_i \in child(v) \land  e_{vu_i} \notin F'_v} TM1[u_i])	   	
\end{equation}

Note that if $v$ is a leaf vertex, $child(v) = \emptyset$. Thus by Equation \ref{eqn1}, $TM1[v] = \emptyset$ for any $T_v$ where $v$ is a leaf vertex.

We next describe the method for computing $TM2[v]$ for any non-root vertex $v$, when for all $T_{u_i}$, $u_i \in child(v)$, $TM1[u_i]$ and $TM2[u_i]$ are already computed. We first define the following.

\begin{definition}
	\textbf{Compatible Edges for $T_v$:} Compatible edges for $T_v$, denoted by $C_v$, is the set of edges $e_{vu_i}$ such that $e_{vu_i} \in F_v$ and $e_{vu_i}$ is non-overlapping with $e_{vp(v)}$.  
\end{definition} 

To compute $TM2[v]$ for $T_v$, the algorithm first computes the set of compatible edges $C_v$ for $T_v$. The algorithm then computes the maximum cardinality subset $C'_v \subseteq C_v$ such that no two edges in $C'_v$ are overlapping with each other. $TM2[v]$ is then computed using the following equation.

\begin{equation} \label{eqn2}
	TM2[v] :=   e_{vp(v)} \cup (\bigcup_{\forall u_i \in child(v) \land e_{vu_i} \in C'_v} TM2[u_i]) \;\cup\; (\bigcup_{\forall u_i \in child(v) \land e_{vu_i} \notin C'_v} TM1[u_i])
\end{equation}

Note that, if $v$ is a leaf vertex, $child(v) = \emptyset$. Thus by Equation \ref{eqn2}, $TM2[v] = \{e_{vp(v)}\}$ for any $T_v$ where $v$ is a leaf vertex. 

For the root vertex $r$ of  $\mathcal{G}$, $p(r) = NULL$, and hence $T_r$ is not defined as the edge $e_{rp(r)}$ does not exist. However, for simplicity, we still compute $TM1[r]$ for $r$ using Equation \ref{eqn1} as the parent of $v$ is not required for computation of $TM1[v]$ for any vertex $v$. The algorithm returns $TM1[r]$ as a maximum 0-1 timed matching for $\mathcal{G}$.

For example, in Figure \ref{exampleTree}, while computing $TM1[b]$ for $T_b$, $TM1[g] = \emptyset$, $TM2[g] = \{e_{bg}\}$ for $T_g$, and $TM1[h] = \emptyset$, $TM2[h] = \{e_{bh}\}$ for $T_h$ where $g$ and $h$ are two children of $b$. Thus by Equation \ref{eqn1}, $TM1[b] = \{e_{bg}\}$ or $\{e_{bh}\}$ and by Equation \ref{eqn2}, $TM2[b] = \{e_{bg}, e_{rb}\}$. Similarly for $T_a$, $TM1[a] = \{e_{ci}, e_{cj}, e_{dl}, e_{dq}\}$ and $TM2[a] = \{e_{ar}, e_{ci}, e_{cj}, e_{dl}, e_{dq}\}$ and for $T_b$, $TM1[b] = \{e_{bh}\}$ and $TM2[b] = \{e_{br}, e_{bg}\}$. Thus, $TM1[r] = \{e_{ar}, e_{ci}, e_{cj}, e_{dl}, e_{dq}, e_{br}, e_{bg}\}$ which is a maximum 0-1 timed matching for the rooted temporal tree shown in Figure \ref{exampleTree}. 

\begin{figure}[t]
	\begin{center}
		\includegraphics[scale=.45]{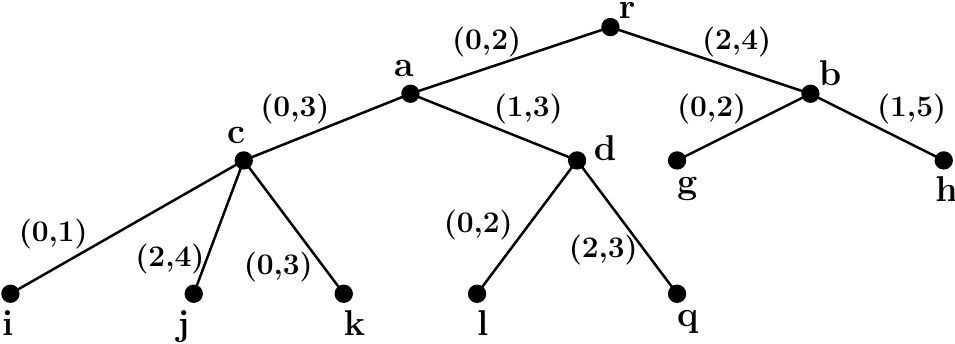}
		\caption{A temporal tree rooted at $r$}
		\label{exampleTree}
	\end{center} 
\end{figure}

Algorithm \ref{alg:dpMatching} describes the pseudocode of the proposed algorithm. The algorithm invokes a function {\em createLevelList} to put all the vertices in $\mathcal{G}$ in different lists, with two vertices put in the same list if and only if their depths in $\mathcal{G}$ are the same. Each $T_v$ is processed in non-increasing order of the depths of vertex $v$ in $\mathcal{G}$, starting from the vertex with the maximum depth. For each $T_v$, the algorithm computes $TM1[v]$ and $TM2[v]$ following Equations \ref{eqn1} and \ref{eqn2} respectively. It returns $TM1[r]$ as a maximum 0-1 timed matching for the given rooted temporal tree where $r$ is the root vertex. While computing $TM1[v]$ and $TM2[v]$ for any $T_v$, Algorithm \ref{alg:dpMatching} internally invokes two other functions. Function $maxNonOverlap(S)$ returns the maximum cardinality non-overlapping subset of edges from a set of edges $S$. Each edge of the rooted temporal tree is associated with a single time interval. Thus, the problem of finding a maximum cardinality non-overlapping subset of edges of a given set of edges reduces to the interval scheduling problem \cite{kleinberg06,cormen09}. Hence, $maxNonOverlap(S)$ finds a maximum cardinality non-overlapping subset of edges from a set of edges $S$ in $O(|S| \log |S|)$ time using the greedy algorithm to solve the interval scheduling problem \cite{kleinberg06,cormen09}.

\begin{algorithm}[H]
	\footnotesize
	\caption{dp0-1TimedMatching}
	\textbf{Input:} $\mathcal{G = (V, E)}$, root vertex $r$\\
	\textbf{Output:} $M \subseteq \mathcal{E}$, a maximum 0-1 timed matching for $\mathcal{G}$
	\label{alg:dpMatching}
	\begin{algorithmic}[1]
		\If{$r = NULL$ \textbf{or} $child(r) = \emptyset$}
		\State $M := \emptyset$; \; \textit{return($M$)}  
		\EndIf
		\State $nList :=$ createLevelList($r$)              \label{level1}
		\For{$level = max\_depth \to 0$} \Comment{$max\_depth = $ maximum depth of a vertex}  \label{level2}
		\While{$(v := nList[level].extractVertex()) != \emptyset$}
		\State $TM1[v] := \emptyset$; \; $TM2[v] := \emptyset$ \label{init}
		\State $F_v := \emptyset$; \; $C_v := \emptyset$
		\ForAll {$u_i \in child(v)$}              \label{sfv}
		\If{$|TM2[u_i]| = |TM1[u_i]|+1$}                      
		\State $F_v := F_v \cup \{e_{vu_i}\}$   \label{efv}
		\EndIf
		\EndFor
		\State $F'_v := maxNonOverlap(F_v)$
		\ForAll{$u_i \in child(v)$}                      \label{stm1}
		\If{$e_{vu_i} \in F'_v$}
		\State $TM1[v] := TM1[v] \cup TM2[u_i]$
		\Else
		\State $TM1[v] := TM1[v] \cup TM1[u_i]$ \label{etm1}
		\EndIf        
		\EndFor
		\If{$p(v) \neq NULL$}            \label{stm2}
		\ForAll{$e_{vu_i} \in F_v$ \textbf{and} $interSect(e_{vu_i}, e_{vp(v)}) = 0$} \label{cc1}
		\State $C_v := C_v \cup \{e_{vu_i}\}$          \label{cc2}
		\EndFor
		\State $C'_v := maxNonOverlap(C_v)$
		\State $TM2[v] := TM2[v] \cup \{e_{vp(v)}\}$              \label{addP}
		\ForAll{$u_i \in child(v)$}
		\If{$e_{vu_i} \in C'_v$}
		\State $TM2[v] := TM2[v] \cup TM2[u_i]$
		\Else
		\State $TM2[v] := TM2[v] \cup TM1[u_i]$ \label{etm2}
		\EndIf        
		\EndFor		            
		\EndIf
		\EndWhile
		\EndFor
		\State $M := TM1[r]$
		\State \textit{return($M$)}
	\end{algorithmic}
	
\end{algorithm}

\noindent The function $interSect(e_{uv}, e_{vw})$ returns $1$ if $e_{uv}, e_{vw}$ are overlapping with each other, returns $0$ otherwise.

\subsection{Proof of Correctness}

\begin{lemma}
	\label{lem:lessTM2}
	Given a rooted temporal tree $\mathcal{G}$, for any $T_v$ in $\mathcal{G}$ such that $p(v) \neq NULL$, if $|TM2[v]| < |TM1[v]| + 1$, then there exists a maximum 0-1 timed matching for $\mathcal{G}$ which does not include the edge $e_{vp(v)}$.
\end{lemma}

\begin{proof}
Let $M$ be a maximum 0-1 timed matching for $\mathcal{G}$. Since $\mathcal{G}$ is a tree, $M$ must include a 0-1 timed matching for $T_v$. Let $M^*_v$ be the 0-1 timed matching of $T_v$ included in $M$. Then two cases are possible:
\begin{itemize}
    \item {\em $M^*_v$ does not include $e_{vp(v)}$}: In this case, $M$ is a maximum 0-1 timed matching of $\mathcal{G}$ that does not include $e_{vp(v)}$. So the lemma holds.
    \item {\em $M^*_v$ includes $e_{vp(v)}$}: In this case, by definition of $TM2[v]$, $|M^*_v| \leq |TM2[v]|$. However, since $M$ is a maximum 0-1 timed matching for $\mathcal{G}$, $|M^*_v| = |TM2[v]|$, as otherwise we can replace $M^*_v$ by $TM2[v]$ in $M$ to get a larger 0-1 timed matching, which is a contradiction. Therefore, $|M^*_v| \leq |TM1[v]|$, as $|TM2[v]| < |TM1[v]| + 1$. Consider $M' = (M \setminus M^*_v) \cup TM1[v]$. Then $M'$ is a 0-1 timed matching of $\mathcal{G}$ and $|M'| \geq |M|$. However, $|M'|$ cannot be greater than $|M|$ as $M$ is a maximum 0-1 timed matching of $\mathcal{G}$. Hence, $|M'| = |M|$. Hence, $M'$ is a maximum 0-1 timed matching of $\mathcal{G}$. Also, $M'$ does not include the edge $e_{vp(v)}$. Hence,  there exists a maximum 0-1 timed matching of $\mathcal{G}$ that does not include the edge $e_{vp(v)}$.
\end{itemize}
	
\end{proof}

\noindent The following lemma follows from Lemma \ref{lem:lessTM2}

\begin{lemma}
\label{lem:lessTM2_TV}
Given a rooted temporal tree $\mathcal{G}$, for any $T_v$ in $\mathcal{G}$ such that $p(v) \neq NULL$, if $|TM2[v]| < |TM1[v]| + 1$, then $TM1[v]$ is a maximum 0-1 timed matching for $T_v$.
\end{lemma}

\begin{lemma}
	\label{lem:maxTM1}
	Suppose that $\mathcal{G = (V, E)}$ is a rooted temporal tree such that $v \in \mathcal{V}$, for each $T_{u_i}$, $u_i \in child(v)$, $TM1[u_i]$ and $TM2[u_i]$ are already computed. Then, Equation \ref{eqn1} correctly computes $TM1[v]$ for $T_v$.  
\end{lemma}

\begin{proof}
    We partition the set $child(v)$ into the following two sets $A$ and $B$:
    \begin{itemize}
        \item $A = \{ u_i \;|\; u_i \in child(v) \;\wedge\; |TM2[u_i]| < |TM1[u_i]| + 1\}$
        \item $B = \{ u_i \;|\; u_i \in child(v) \;\wedge\; |TM2[u_i]| = |TM1[u_i]| + 1\}$
    \end{itemize}
    Note that $A \cup B = child(v)$ \footnote{ For any $T_{u_i}$, if $|TM2[u_i]| > |TM1[u_i]| + 1$, we can get a 0-1 timed matching $M'' = TM2[u_i] \setminus \{e_{p(u_i)u_i}\}$ for $T_{u_i}$ such that $M''$ does not include edge $e_{p(u_i)u_i}$ and $|TM1[u_i]| < |M''|$. According to definition of $TM1[u_i]$, this is impossible.} and $A \cap B = \emptyset$. Also, by Definition \ref{def:feasible}, $F_v = \{e_{vu_i}\;|\; u_i \in B\}$. Hence, for all nodes $u_i \in A$, $e_{vu_i} \notin F_v$.
    
    By definition, $TM1[v]$ is a maximum 0-1 timed matching for $T_v$ when edge $e_{vp(v)}$ is not included. Since $T_v$ is a temporal tree, for each $u_i \in child(v)$, $TM1[v]$ should include the maximum possible sized 0-1 timed matching for  $T_{u_i}$, without violating the properties of a 0-1 timed matching. Two cases are possible for including the maximum 0-1 timed matching for $T_{u_i}$ for a node $u_i \in child(v)$:
    \begin{itemize}
        \item{Case 1: $u_i \in A$:} In this case, From Lemma \ref{lem:lessTM2_TV}, $TM1[u_i]$ is a maximum 0-1 timed matching for $T_{u_i}$.  Also, since $TM1[u_i]$ does not include the edge $e_{vu_i}$, inclusion of $TM1[u_i]$ does not violate the 0-1 timed matching property of $TM1[v]$ irrespective of any choice made for including either $TM1[u_k]$ or $TM2[u_k]$ for any other $T_{u_k}$, $u_k \in child(v)$. Hence, in this case, $TM1[u_i]$ should be included in $TM1[v]$.
        \item{Case 2: $u_i \in B$:} In this case, $|TM2[u_i]| > |TM1[u_i]|$ and hence, $TM2[u_i]$ should be included in $TM1[v]$ instead of $TM1[u_i]$ if possible. However, including $TM2[u_k]$ for all $T_{u_k}$, $u_k \in B$ may violate the 0-1 timed matching property as $e_{vu_x}, e_{vu_y}$ can be overlapping for two distinct trees $T_{u_x}, T_{u_y}$, such that $u_x, u_y \in B$. Since $|TM2[u_i]| = |TM1[u_i]| + 1$ in this case, the maximum number of edges can be included in $TM1[v]$ from the 0-1 timed matchings of trees $T_{u_i}$, $u_i \in B$ in the following way: (i) include $TM2[u_r]$ in $TM1[v]$ for the maximum possible number of trees $T_{u_r}$, $u_r \in B$ such that there are no overlapping edges, and (ii) include $TM1[u_s]$ in $TM1[v]$ for the remaining trees $T_{u_s}$, $u_s \in B$.
    \end{itemize}
    
     Equation \ref{eqn1} includes $TM1[u_i]$ for all $u_i \in A$  in $TM1[v]$ (as required in Case 1 above for nodes in $A$) as it includes $TM1[u_i]$ for all trees $T_{u_i}$, $u_i \in child(v)$ such that $e_{vu_i} \notin F_v'$ (since for all $u_i \in A$, $e_{vu_i} \notin F_v$ and hence $e_{vu_i} \notin F'_v$). Also, $F_v'$ is the maximum cardinality subset of $F_v$ with no overlapping edges, and Equation \ref{eqn1} also includes $TM2[u_i]$ for all trees $T_{u_i}$, such that $e_{vu_i} \in F_v'$. Hence, since $F_v = \{e_{vu_i}\,|\, u_i \in B\}$,  Equation \ref{eqn1} includes (as required in Case 2 above for nodes in $B$) $TM2[u_r]$ for maximum possible number of trees $T_{u_r}$, $u_r \in B$. Also, for all other trees $T_{u_s}$, $u_s \in B$ it includes $TM1[u_s]$ as $u_s \notin F_v'$ (as required in Case 2 above).  Hence Equation \ref{eqn1} correctly chooses the 0-1 timed matchings for all $u_i \in child(v)$ for maximizing the size of $TM1[v]$. Hence Equation \ref{eqn1} computes $TM1[v]$ correctly. 
\end{proof}

\noindent The following lemma can be proved using similar arguments used to prove Lemma \ref{lem:maxTM1}.

\begin{lemma}
	\label{lem:maxTM2}
	Suppose that $\mathcal{G = (V, E)}$ is a rooted temporal tree such that $v \in \mathcal{V}$, for each $T_{u_i}$, $u_i \in child(v)$, $TM1[u_i]$ and $TM2[u_i]$ are already computed. Then, Equation \ref{eqn2} correctly computes $TM2[v]$ for $T_v$.  
\end{lemma}

\begin{theorem}
	Algorithm \ref{alg:dpMatching} correctly computes a maximum 0-1 timed matching for a given rooted temporal tree $\mathcal{G = (V, E)}$. 
	\label{thm:correctness}
\end{theorem}

\begin{proof}
	We first prove that the algorithm correctly computes $TM1[v]$ and $TM2[v]$ for each $T_v$ where $v \in \mathcal{V}$. We prove this by induction on the height of the temporal tree $T_v$. 
	
	\textit{\textbf{Base Case:}} {\em The height of $T_v$ is 1}: Algorithm \ref{alg:dpMatching}, in Line \ref{init}, initializes the values of $TM1[v]$ and $TM2[v]$ to $\emptyset$ for any $T_v$. As height of $T_v$ is 1, $v$ is a leaf vertex. Thus, the computed $TM1[v] = \emptyset$ and $TM2[v] = \{e_{vp(v)}\}$ (added at Line \ref{addP}) are correct.
	
	\textit{\textbf{Inductive Step:}} Assume that Algorithm \ref{alg:dpMatching} correctly computes $TM1[x]$ and $TM2[x]$ for any temporal tree $T_x$ with height up to $l$. We need to prove that Algorithm \ref{alg:dpMatching} correctly computes $TM1[v]$ and $TM2[v]$ for temporal tree $T_v$ with height $l+1$.
	Since the height of $T_v$ is $l+1$, the depth of $v$ is at most $max\_depth - l$, where $max\_depth$ is the maximum depth of any vertex in $\mathcal{G}$. Then, for any  $w \in child(v)$, the depth of $w$ is at most $max\_depth - l + 1$, and hence, the height of $w$ is at most $l$. Since the algorithm processes vertices in non-increasing order of their depths (Lines \ref{level1} and \ref{level2} in Algorithm \ref{alg:dpMatching}), when $TM1[v]$ and $TM2[v]$ are computed, $TM1[w]$ and $TM2[w]$ are already computed correctly for all $T_w, w \in child(v)$ by the induction hypothesis. Hence, by Lemma \ref{lem:maxTM1} and \ref{lem:maxTM2}, $TM1[v]$ and $TM2[v]$ are correctly computed using Equation \ref{eqn1} (Lines \ref{sfv} to \ref{etm1} in Algorithm 1) and Equation \ref{eqn2} (Lines \ref{sfv} to \ref{efv}, and Lines \ref{stm2} to \ref{etm2} in Algorithm 1). 
	
	Thus, Algorithm \ref{alg:dpMatching} correctly computes $TM1[r]$ where $r$ is the root vertex of $\mathcal{G}$ and returns it as a maximum 0-1 timed matching for $\mathcal{G}$. Hence the theorem holds.
\end{proof}

\begin{theorem}
	The time complexity of Algorithm \ref{alg:dpMatching} is $O(n\log n)$.
	\label{thm:runtime}
\end{theorem}

\begin{proof}
	Algorithm \ref{alg:dpMatching} stores the vertices in the rooted temporal tree $\mathcal{G = (V, E)}$ in different lists according to their depth in $\mathcal{G}$ such that two vertices with the same depth are in the same list. This can be done in $O(n)$ time by traversing the rooted temporal tree in level order where $|\mathcal{V}| = n$. 
	
	While computing $TM1[v_i]$ and $TM2[v_i]$ for $T_{v_i}$ where $v_i \in \mathcal{V}$,  $TM1[v_j]$ and $TM2[v_j]$ of $T_{v_j}$ are available for all $v_j \in child(v_i)$. In this scenario, as intermediate steps, Algorithm \ref{alg:dpMatching} computes $F_{v_i}$ and $C_{v_i}$ at Lines \ref{sfv} to \ref{efv} and Lines \ref{cc1} to \ref{cc2} respectively. The number of edges incident on any vertex $v_i$ is $d(v_i)$ where $d(v_i)$ is the degree of vertex $v_i$ in $\mathcal{G}_U$, and each edge is associated with a single time interval. Thus, computing $F_{v_i}$ and $C_{v_i}$ take $O(d(v_i))$ time. Since the maximum size of $F_{v_i}$ and $C_{v_i}$ are $O(d(v_i))$, finding the maximum cardinality non-overlapping subsets $F'_{v_i}$ and $C'_{v_i}$ take $O(d(v_i) \log d(v_i))$ time using the function $maxNonOverlap$. Thus, the computation of $TM1[v_i]$ and $TM2[v_i]$ for any $T_{v_i}$ takes $O(d(v_i) \log d(v_i))$ time. Hence, total time taken for computing $TM1[v_i]$ and $TM2[v_i]$ for all $T_{v_i}$, $v_i \in \mathcal{V}$, in the temporal tree is $\sum_{i=1}^n O(d(v_i) \log d(v_i))$ $=$ $O(\sum_{i=1}^n d(v_i) \log d(v_i))$. Since for any $v_i \in \mathcal{V}$, $d(v_i) < n$, $\sum_{i=1}^n d(v_i) \log d(v_i)  \leq  \sum_{i=1}^n d(v_i) \log n$. Hence, $\sum_{i=1}^n d(v_i) \log d(v_i) \leq  \log n \sum_{i=1}^n d(v_i)$. Since $\mathcal{G}_U$ is a tree, $\sum_{i = 1}^n d(v_i) = 2(n-1)$. Hence, $\sum_{i=1}^n d(v_i) \log d(v_i) \leq 2(n - 1) \log n$. Hence, the time complexity of Algorithm \ref{alg:dpMatching} is $O(n \log n)$.
\end{proof}

\section{Finding a Maximum 0-1 Timed Matching for Temporal Graphs}
\label{bipartite}
In this section, we address the problem of finding a maximum 0-1 timed matching for a given temporal graph which is not a tree. We first analyse the computational complexity of the problem. After that, we analyse the approximation hardness of the problem of finding a maximum 0-1 timed matching for a given temporal graph when multiple time intervals are associated with each edge. Finally, we propose an approximation algorithm for the problem.

\subsection{Complexity of Finding Maximum 0-1 Timed Matching in Temporal Graphs}
 We first prove that the decision version of the problem of finding a maximum 0-1 timed matching for temporal graphs when each edge is associated with multiple time intervals, referred to as D-MAX-0-1-TM-MULT, is in NP. The D-MAX-0-1-TM-MULT problem is defined as follows.

\begin{definition}
	\textbf{D-MAX-0-1-TM-MULT:} Given a temporal graph $\mathcal{G = (V, E)}$ with lifetime $\mathcal{T}$, where each edge in $\mathcal{E}$ is associated with multiple time intervals, and a positive integer $f$, does there exist a 0-1 timed matching $M$ for $\mathcal{G}$ such that $|M| = f$?
\end{definition}
	
\begin{lemma}
	\label{lem:tmtMultNP}
	The D-MAX-0-1-TM-MULT problem is in NP.
\end{lemma}	 

\begin{proof}
	Consider a certificate $\langle \langle \mathcal{G = (V,E)}, f \rangle, M\rangle$, where $\mathcal{G}$ is a temporal graph with lifetime $\mathcal{T}$ such that each edge is associated with multiple time intervals, $f$ is a given positive integer, and $M$ is a given set of edges. We consider one edge $e_{uv} \in M$ at a time and compare the associated time intervals of $e_{uv}$ with associated time intervals of all the other edges in $M$. We perform this check for all the edges in $M$ to find if there exists any pair of edges overlapping with each other. As discussed in Section \ref{model}, the maximum number of intervals associated with an edge is $\lfloor \frac{\mathcal{T}}{2} \rfloor$. Hence, this checking can be done in polynomial time. Whether $|M| =  f$ and $M \subseteq \mathcal{E}$ can also be easily checked in polynomial time. Hence, the problem of finding a maximum 0-1 timed matching for temporal graphs when each edge is associated with multiple time intervals is in NP.
\end{proof}

We have already proved in Theorem \ref{thm:NPTree} that the problem of finding a maximum 0-1 timed matching for a given rooted temporal tree when each edge is associated with multiple time intervals is NP-complete. Thus, from Theorem \ref{thm:NPTree} and Lemma \ref{lem:tmtMultNP}, we get the following result:

\begin{corollary}
The problem of finding a maximum 0-1 timed matching for temporal graphs when each edge is associated with multiple time intervals is NP-complete.
\end{corollary}

Next, we investigate the computational complexity of the problem of finding a  maximum 0-1 timed matching for temporal graphs when each edge is associated with a single time interval. We prove that this problem is also NP-complete. In order to show that, we prove that this problem is NP-complete even when the given temporal graph is a degree-bounded bipartite temporal graph such that degree of each vertex is bounded by $3$ and each edge is associated with a single time interval. We refer to this problem as D-MAX-0-1-TMBD3B-1 problem. We first define the D-MAX-0-1-TMBD3B-1 problem. In the rest of this paper, we refer to a bounded degree bipartite temporal graph when degree of each vertex is bounded by $3$ as BDG3B.

\begin{definition}
 \textbf{D-MAX-0-1-TMBD3B-1:} Given a BDG3B $\mathcal{G = (V, E)}$ with lifetime $\mathcal{T}$, where each edge in $\mathcal{E}$ is associated with a single time interval, and a positive integer $z$, does there exist a 0-1 timed matching $M$ for $\mathcal{G}$ such that $|M| = z$?
\end{definition}

Next, we prove NP-completeness of the D-MAX-0-1-TMBD3B-1 problem by showing that there is a polynomial time reduction from the 2P2N-3SAT problem \cite{Caragiannis19} to the D-MAX-0-1-TMBD3B-1 problem. The 2P2N-3SAT problem is known to be NP-complete \cite{Caragiannis19}. The 2P2N-3SAT problem is defined as follows.

\begin{definition}
 \textbf{2P2N-3SAT:} Let $U=\{v_1, v_2, \cdots, v_m\}$ be a set of $m$ boolean variables and let $\Psi$ be a 3-CNF formula with $d$ clauses $C_1, C_2, \cdots, C_d$ such that each variable occurs exactly twice positively and twice negatively in $\Psi$. Does there exist a truth assignment which satisfies $\Psi$? 
\end{definition}

\begin{theorem}
 The D-MAX-0-1-TMBD3B-1 problem is NP-complete.
 \label{thm:bdg3}
\end{theorem}

\begin{proof}
We first show that the D-MAX-0-1-TMBD3B-1 problem is in NP. Consider a certificate $\langle \langle \mathcal{G = (V,E)},$ $z \rangle, M\rangle$, where $\mathcal{G}$ is a BDG3B with lifetime $\mathcal{T}$ such that each edge in $\mathcal{E}$ is associated with a single time interval, $z$ is a given integer, and $M$ is a given set of edges. We consider one edge $e_{uv} \in M$ at a time and compare the associated time interval of $e_{uv}$ with associated time intervals of all the other edges in $M$. We perform this check for all the edges in $M$ to find if there exist any edges overlapping with each other. This checking can be done in polynomial time. Whether $|M| =  z$ and $M \subseteq \mathcal{E}$ can also be easily checked in polynomial time. Hence, the D-MAX-0-1-TMBD3B-1 problem is in NP.

Next, we prove that, there is a polynomial time reduction from the 2P2N-3SAT problem to the D-MAX-0-1-TMBD3B-1 problem. Let $\langle U, \Psi \rangle$ be an instance of the 2P2N-3SAT problem where $U=\{v_1, v_2, \cdots, v_m\}$ be a set of $m$ variables and $\Psi$ be a 3-CNF formula with $d = \frac{4m}{3}$ clauses $C_1, C_2, \cdots, C_d$ such that each clause $C_i$ consists of exactly $3$ literals and each variable occurs exactly twice positively and twice negatively in $\Psi$. A literal is a variable or negation of a variable in $U$. Without loss of generality, we assume that each clause in $\Psi$ consists of distinct literals. From this instance of the 2P2N-3SAT problem, we construct an instance of the D-MAX-0-1-TMBD3B-1 problem as follows.
\begin{itemize}
    \item We construct a temporal graph $\mathcal{G = (V, E)}$ as follows:
    \begin{enumerate}
        \item We add a vertex $c_i$ to a set $A_1$ for each clause $C_i$ in $\Psi$. Thus,
                    
            $A_1 := \{c_i\,|\, \forall C_i \in \Psi\}$
          
            We add a vertex $a_i$ to a set $A_2$ for each variable $v_i$ in $U$. Thus,
                    
            $A_2 := \{a_i\,|\, \forall v_i \in U\}$ 
            
            We add two vertices $b_i, \Bar{b}_i$ to a set $B$ for each variable $v_i$ in $U$. Thus,
                
            $B := \{b_i, \Bar{b}_i\,|\, \forall v_i \in U\}$
            
            Then the set of vertices $\mathcal{V}$ in the temporal graph $\mathcal{G}$ is
            
            $\mathcal{V} := A \cup B$ where $A := A_1 \cup A_2$
        \item We add an edge between vertex $c_i$ and vertex $b_j$ which exists for time interval $(i-1, i)$, if $v_j$ is a literal in clause $C_i$. Thus,
        
            $\mathcal{E}_{cb}$ $:= \{e(c_i, b_j, (i-1, i))\,|\, \forall v_j \in U, C_i \in \Psi,$ $v_j$ is a literal in $C_i\}$
            
            We add an edge between vertex $c_i$ and vertex $\bar{b}_j$ which exists for time interval $(i-1, i)$, if $\bar{v}_j$ is a literal in clause $C_i$. Thus,
            
            $\mathcal{E}_{c\bar{b}}$ $:= \{e(c_i, \bar{b}_j, (i-1, i))\,|\,\forall v_j \in U, C_i \in \Psi,$ $\bar{v}_j$ is a literal in $C_i\}$
            
            We add an edge between each vertex $a_i$ and vertex $b_i$ which exists for time interval $(0, d)$. Thus,
            
            $\mathcal{E}_{ab}$ $:= \{e(a_i, b_i, (0, d))\,|\, \forall a_i, b_i \in \mathcal{V}\}$
            
            We add an edge between each vertex $a_i$ and vertex $\bar{b}_i$ which exists for time interval $(0, d)$. Thus,
            
            $\mathcal{E}_{a\bar{b}}$ $:= \{e(a_i, \bar{b}_i, (0, d))\,|\, \forall a_i, \bar{b}_i \in \mathcal{V}\}$
            
            Then the set of edges $\mathcal{E}$ in the temporal graph $\mathcal{G}$ is
            
            $\mathcal{E} := \mathcal{E}_{cb} \cup \mathcal{E}_{c\bar{b}} \cup \mathcal{E}_{ab} \cup \mathcal{E}_{a\bar{b}}$

        \item $\mathcal{T} := d$
    \end{enumerate}            
    \item $z := d+m$
\end{itemize}
 As there are exactly 3 literals in each clause $C_j$ and for each variable $v_i$, literals $v_i$ and $\bar{v}_i$ are present in exactly two clauses each, $\mathcal{G}$ is a BDG3B. Any edge $e(u, v, (s, f)) \in \mathcal{E}$ is also denoted as $e_{uv}$ when the time interval for which this edge exists is not important. Figure \ref{fig:bounded} shows the construction of a BDG3B from an instance of the 2P2N-3SAT problem. 

\begin{figure}
 \begin{center}
  \includegraphics[scale=.52]{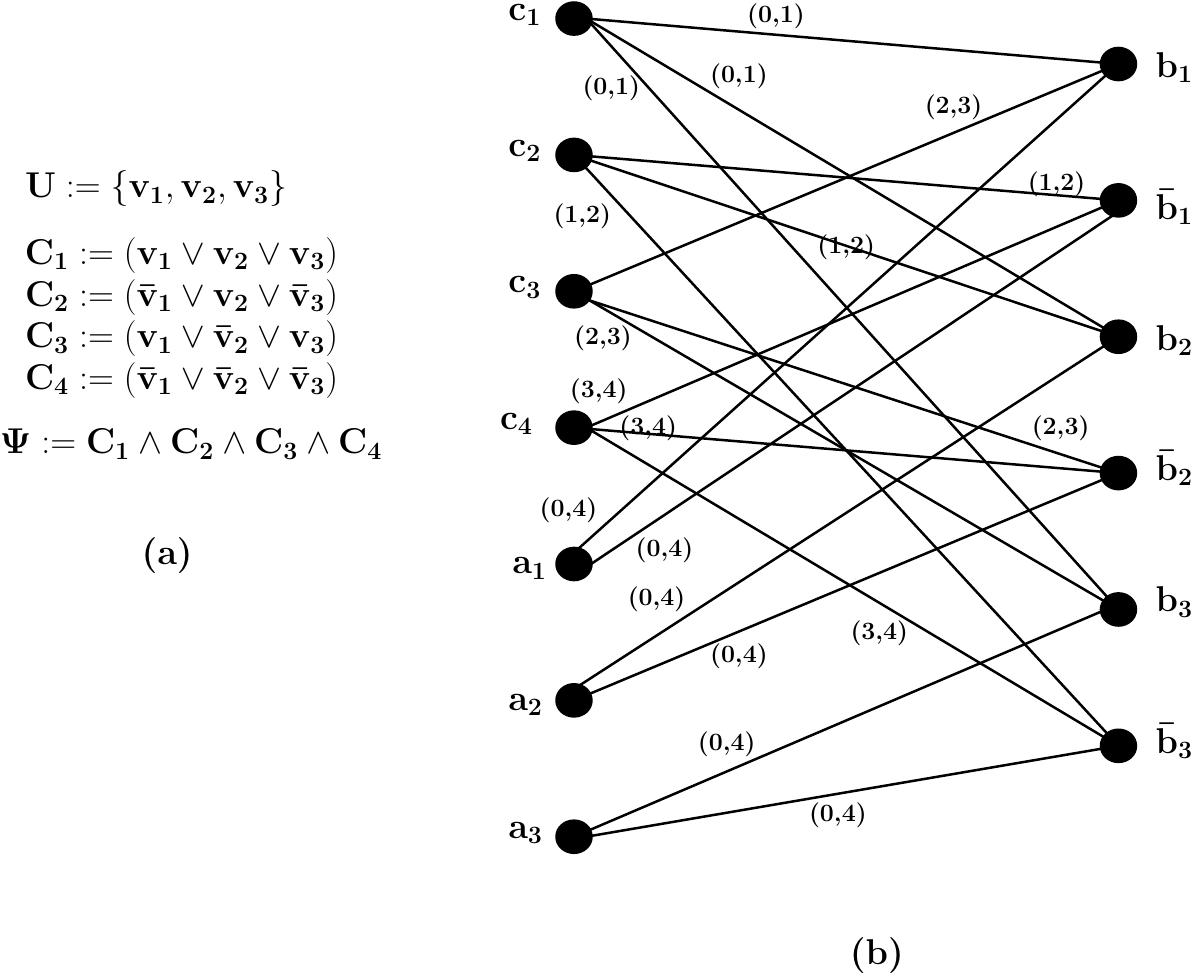}
  \caption{Construction of a BDG3B from the input of an instance of the 2P2N-3SAT problem.}
  \label{fig:bounded}
 \end{center} 
 \end{figure}

We first prove that if there is a solution for the D-MAX-0-1-TMBD3B-1 problem, then there is a truth assignment which satisfies $\Psi$. For a 0-1 timed matching $M$, $|M| = z$, for $\mathcal{G}$, we construct a truth assignment $\mathcal{S}$ for $\Psi$ as follows. 
\begin{enumerate}
    \item For each $e_{a_ib_i} \in M$, assign $v_i := \mathbf{false}$. \label{R1}
    \item For each $e_{a_i\bar{b}_i} \in M$, assign $v_i := \mathbf{true}$. \label{R2}
\end{enumerate}

We first show that, for each vertex $v \in A$, there is exactly one edge incident on $v$ is in $M$. According to our construction, $\mathcal{G}$ is bipartite and $|A| = m+d$. Again, any two edges incident on the same vertex in $A$ are overlapping with each other. Thus, no two edges incident on the same vertex in $A$ can be in $M$. As $|M| = m+d$ and $|A| = m+d$, for each vertex $v \in A$, there is exactly one edge incident on $v$ in $M$. This shows that, for each vertex $a_i \in A_2$ there is exactly one edge incident on $a_i$ in $M$. This implies that, each variable $v_i \in U$ is assigned to a truth value and no variable in $U$ is assigned to both $\mathbf{true}$ and $\mathbf{false}$.  

Next we show that $\mathcal{S}$ satisfies all the clauses in $\Psi$. For any clause $C_k$, we show that at least one literal  $v_i \in C_k$ or $\bar{v}_i \in C_k$ is satisfied. We have shown that, for each vertex $c_k \in A_1$, there is one edge incident on $c_k$ in $M$. There are two possible cases:

\begin{enumerate}[I.]
    \item {\em Some edge $e_{c_kb_i}$ incident on $c_k$ is in $M$:} This implies that if variable $v_i$ is set to \textbf{true}, the clause $C_k$ is satisfied. According to the truth assignment in $\mathcal{S}$, if $e_{a_i\bar{b}_i}$ is in $M$, then $v_i$ is set to \textbf{true}. Again we have already proved that, for each $a_i \in A_2$, there is an edge incident on $a_i$ in $M$. If $e_{a_ib_i}$ is in $M$, then $e_{c_kb_i}$ cannot be in $M$ because these two edges are overlapping with each other. Hence, $e_{a_i\bar{b}_i} \in M$ and $C_k$ is satisfied.
    
    \item {\em Some edge $e_{c_k\bar{b}_i}$ incident on $c_k$ is in $M$:} This implies that if variable $v_i$ is set to \textbf{false}, the clause $C_k$ is satisfied. According to the truth assignment in $\mathcal{S}$, if $e_{a_ib_i}$ is in $M$, then $v_i$ is set to \textbf{false}. Again we have already proved that, for each $a_i \in A_2$, there is an edge incident on $a_i$ in $M$. If $e_{a_i\bar{b}_i}$ is in $M$, then $e_{c_k\bar{b}_i}$ cannot be in $M$ because these two edges are overlapping with each other. Hence, $e_{a_ib_i} \in M$ and $C_k$ is satisfied.
\end{enumerate}

Thus, each clause $C_k$ in $\Psi$ is satisfied by $\mathcal{S}$ and each variable in $U$ is assigned to either \textbf{true} or \textbf{false}. Hence, $\mathcal{S}$ is a truth assignment which satisfies $\Psi$. 

Next, we prove that if there is no solution for the instance of the D-MAX-0-1-TMBD3B-1 problem, then there is no truth assignment which satisfies $\Psi$. We prove this by showing that if there is a truth assignment which satisfies $\Psi$, then there is a solution for the instance of the D-MAX-0-1-TMBD3B-1 problem. For a satisfying truth assignment $\mathcal{S}$ of $\Psi$, we construct a 0-1 timed matching $M$ of size $m+d$ for $\mathcal{G}$ as follows. 

\begin{enumerate}
	\item For any variable $v_i \in U$, if $v_i = \mathbf{false}$ is in $\mathcal{S}$, we include $e_{a_ib_i}$ in $M$.
	\item For any variable $v_i \in U$, if $v_i = \mathbf{true}$ is in $\mathcal{S}$, we include $e_{a_i\bar{b}_i}$ in $M$.
    \item For any clause $C_l$ in $\Psi$, we select any one literal $v_i$ in $C_l$ such that $v_i$ evaluates to $\textbf{true}$ following assignment $\mathcal{S}$, then we include $e_{c_lb_i}$ in $M$. If the selected literal is $\bar{v}_i$ in $C_l$ such that $\bar{v}_i$ evaluates to $\textbf{true}$ following assignment $\mathcal{S}$, then we include $e_{c_l\bar{b}_i}$ in $M$. Note that, for each clause $C_l$, we include exactly one such edge incident on $c_l$ in $M$.
\end{enumerate}

We first prove that $M$ is a 0-1 timed matching for $\mathcal{G}$. We prove this by contradiction. Assume that, $M$ is not a 0-1 timed matching for $\mathcal{G}$. Then there are at least two edges in $M$, which are overlapping with each other. According to the construction of $\mathcal{G}$, any two edges incident on two different vertices in $A_1$ are non-overlapping with each other. $M$ includes only one edge incident on each vertex in $A_1$. As each variable $v_i$ is assigned to either $\mathbf{true}$ or $\mathbf{false}$, only one edge incident on a vertex in $A_2$ is included in $M$. Thus, two edges can overlap in following two cases:
\begin{enumerate}[I.]
    \item There is some vertex $b_i \in B$ such that $e_{a_ib_i}$ and some edge $e_{c_jb_i}$ are both included in $M$. According to the construction of $M$, $e_{a_ib_i} \in M$ implies that $v_i$ is assigned to $\mathbf{false}$. Again $e_{c_jb_i} \in M$ implies that, $v_i = \mathbf{true}$ and $v_i \in C_j$. Thus, $e_{a_ib_i}$ and $e_{c_jb_i}$ both can be included in $M$ only when $v_i$ is assigned to both $\mathbf{true}$ and $\mathbf{false}$. This is a contradiction.
    \item There is some vertex $\bar{b}_i \in B$ such that $e_{a_i\bar{b}_i}$ and some edge $e_{c_j\bar{b}_i}$ are both included in $M$. According to the construction of $M$, $e_{a_i\bar{b}_i} \in M$ implies that $v_i$ is assigned to $\mathbf{true}$. Again $e_{c_j\bar{b}_i} \in M$ implies that, $v_i = \mathbf{false}$ and $\bar{v}_i \in C_j$. Thus, $e_{a_i\bar{b}_i}$ and $e_{c_j\bar{b}_i}$ both can be included in $M$ only when $v_i$ is assigned to both $\mathbf{true}$ and $\mathbf{false}$. This is also a contradiction.
\end{enumerate}

Next, we prove that $|M| = m+d = z$. We prove this by contradiction. There can be following two possible cases:
\begin{enumerate}[I.]
	\item  {\em $|M| > m+d$:} According to the construction of $M$, for each vertex $v \in A$, only one edge incident on $v$ is included in $M$. As $\mathcal{G}$ is bipartite and $|A| = m+d$, $|M| > m+d$ is impossible.
	\item {\em $|M| < m+d$:} $\mathcal{S}$ satisfies $\Psi$. There are $d$ clauses in $\Psi$. Hence, according to the construction of $M$, one edge incident on each vertex in $A_1$ is included in $M$. Thus $|M| < m+d$ implies that, there is at least one vertex $a_i \in A_2$ such that no edge incident on $a_i$ is included in $M$. This implies that, there is a variable $v_i \in U$ which is not assigned to any truth value in $\mathcal{S}$. This is a contradiction.
\end{enumerate}
 Hence, $|M| = m+d = z$. This completes the proof.  
\end{proof}

\noindent The following result directly follows from Theorem \ref{thm:bdg3}. 

\begin{corollary}
The problem of finding a maximum 0-1 timed matching for temporal graphs when each edge is associated with a single time interval is NP-complete.
\end{corollary}

\subsection{Approximation Hardness of Finding Maximum 0-1 Timed Matching for Temporal Graphs}
Next, we study the approximation hardness of the problem of finding a maximum 0-1 timed matching for a temporal graph when multiple time intervals are associated with each edge. In order to show inapproximability of the problem of finding maximum 0-1 timed matching for temporal graphs, we prove that there is no $\frac{1}{n^{1-\epsilon}}$, for any $\epsilon > 0$, factor approximation algorithm for finding a maximum 0-1 timed matching for a rooted temporal tree when each edge is associated with multiple time intervals unless NP = ZPP. We refer to the problem of finding a maximum 0-1 timed matching for a rooted temporal tree when multiple time intervals are associated with each edge as the MAX-0-1-TMT-MULT problem. We prove this by showing that there exists an approximation preserving reduction from the maximum independent set (MAX-IS) problem to the MAX-0-1-TMT-MULT problem. It is already known that there is no $\frac{1}{n^{1-\epsilon}}$, for any $\epsilon > 0$, factor approximation algorithm for the MAX-IS problem unless NP = ZPP \cite{Hastad96,Berman99}. We define these two problems first.

\begin{definition}
 \textbf{MAX-0-1-TMT-MULT:} Given a rooted temporal tree $\mathcal{G = (V, E)}$ with lifetime $\mathcal{T}$, where each edge in $\mathcal{E}$ is associated with multiple time intervals, find a maximum sized 0-1 timed matching $M$ for $\mathcal{G}$.
\end{definition}

 \begin{definition}
  \textbf{MAX-IS:} Given a static graph $G = (V, E)$, find a maximum sized set $I \subseteq V$ such that no two vertices in $I$ are connected by an edge in $E$.
 \end{definition}

 \begin{theorem}
 There exists an approximation preserving reduction from the MAX-IS problem to the MAX-0-1-TMT-MULT problem.  
  \label{thm:hardApprox}
 \end{theorem}

 \begin{proof}
 Consider an instance $\langle G = (V, E) \rangle$ of the MAX-IS problem where $G$ is a static graph with $|V| = n$, and $|E| = m$. For our reduction, we assume that each edge $(u,v) \in E$ is labelled with a distinct integer $a_{uv}$ between $0$ to $m-1$. We assume that the set $V_0$ of $0$ degree vertices, $V_0 \subseteq V$, is given and $|V_0| = n_0$. If $V_0$ is not given, it can be easily computed in polynomial time. We also assume that, each vertex $v \in V_0$ is labelled with a distinct integer $a_v$, between $m$ to $m+n_0-1$. For each vertex $v \in V$, the {\em neighbouring vertex set of $v$, $N_v$}, includes each vertex $w_i$ such that edge $(v, w_i) \in E$. Note that, for any vertex $u \in V_0$, $N_u = \emptyset$. 

From the given instance of the MAX-IS problem, we construct an instance of the MAX-0-1-TMT-MULT problem as follows.
 \begin{itemize}
  \item We construct the temporal graph $\mathcal{G = (V, E)}$ as follows
  \begin{itemize}
     \item We add a vertex $\nu(v)$ for each vertex $v \in V$. Additionally add a vertex $r$. Thus,
     
     $\mathcal{V} := \{\nu(v)\;|\; \forall v \in V\} \cup \{r\}$.
     \item We add an edge between each $\nu(v)$ added to $\mathcal{V}$ for each vertex $v \in V_0$ to $r$. This edge exists for the time interval $(a_v, a_v+1)$. Thus,
    
     $\mathcal{E}_{deg0}$ $:=$ $\{ e(\nu(v), r, (a_v, a_v+1))$ $|$ $\forall v \in V_0\}$
    
     \item We add an edge between each $\nu(v)$ added to $\mathcal{V}$ for each vertex $v \in V \setminus V_0$ and $r$. This edge exists for the time intervals $(a_{vw_1}, a_{vw_1}+1), (a_{vw_2}, a_{vw_2}+1), \cdots, (a_{vw_k}, a_{vw_k}+1)$ where each $w_i \in N_v$, $k$ is the degree of $v$ ($deg(v)$), and each $a_{vw_i}$ is the integer labelling edge $(v, w_i)$ incident on $v$. Thus,
     
     $\mathcal{E}_{degk}$ $:=$ $\{e(\nu(v), r, (a_{vw_1}, a_{vw_2}+1), \cdots, (a_{vw_k}, a_{vw_k}+1))$ $|$ $\forall v \in (V \setminus V_0),\, \forall w_i \in N_v,\, k = deg(v)\}$
     \item Then the set of edges $\mathcal{E}$ in the temporal graph $\mathcal{G}$ is
     
     $\mathcal{E} := \mathcal{E}_{deg0} \cup \mathcal{E}_{degk}$

     \item Lifetime $\mathcal{T} = m+n_0$
 \end{itemize}
\end{itemize}

 \begin{figure}
 \begin{center}
  \includegraphics[scale=.55]{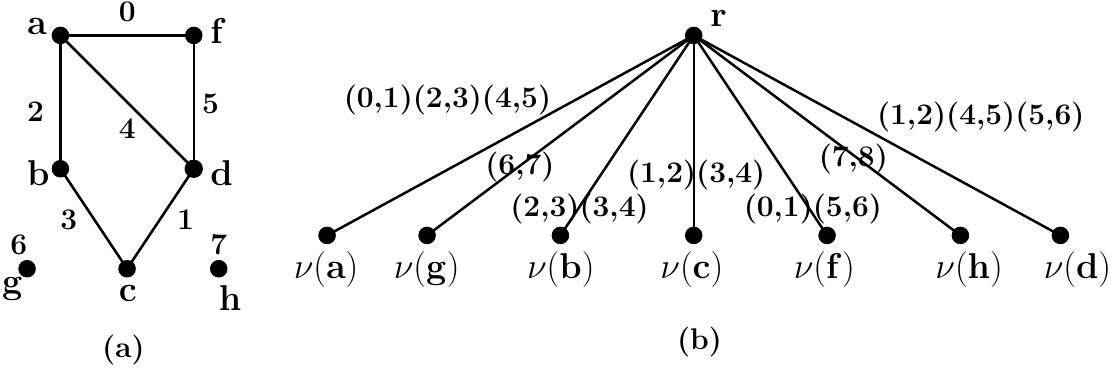}
  \caption{(a) A static graph where edges are labelled by integers, g and h are two $0$ degree vertices which are also labelled with integers (b) Corresponding temporal tree}
  \label{fig:indSet}
 \end{center} 
 \end{figure}

There is an edge between each vertex in $\mathcal{V} \setminus \{r\}$ and $r$, and there is no edge between any other vertices. Thus, $\mathcal{G}$ is a temporal tree rooted at $r$ (we choose $r$ as the root of $\mathcal{G}$). Any edge $e(u, v, (s_1, f_1), (s_2, f_2), \cdots, (s_i, f_i)) \in \mathcal{E}$ is also denoted as $e_{uv}$ when the time intervals for which this edge exists are not important. Figure \ref{fig:indSet} shows the construction of a rooted temporal tree from a given static graph where each edge in the temporal tree is associated with multiple time intervals. 

For a 0-1 timed matching $M$ for $\mathcal{G}$, we construct a solution $I$ for the MAX-IS problem on $G$ such that $|I| = |M|$ as follows. Consider the set of vertices $I = \{v\,|\,e_{r\nu(v)} \in M \}$. We show that $I$ is an independent set for $G$ and $|I| = |M|$. As for each edge in $M$ a vertex in $I$ is selected, $|I| = |M|$. We prove that $I$ is an independent set for $G$ by contradiction. Assume that $I$ is not an independent set for $G$. This implies that there are at least two vertices $u, v  \in I$ such that $(u, v) \in E$. As $u, v \in I$, both $e_{\nu(u)r}$ and $e_{\nu(v)r}$ are in $M$. As $(u,v) \in E$, $(u, v)$ is labelled with integer $a_{uv}$, according to our construction, the time interval $(a_{uv}, a_{uv}+1)$ is associated with both $e_{\nu(u)r}$ and $e_{\nu(v)r}$, and both are incident on $r$. This implies that $M$ is not a 0-1 timed matching for $\mathcal{G}$. This is a contradiction. This completes the proof of the theorem. 
\end{proof}

Therefore, from Theorem \ref{thm:hardApprox} and using the approximation hardness result of the MAX-IS \cite{Hastad96,Berman99} problem, we conclude that, there is no $\frac{1}{n^{1-\epsilon}}$, for any $\epsilon > 0$, factor approximation algorithm for the MAX-0-1-TMT-MULT problem unless NP = ZPP. From this result for a rooted temporal tree, we get the following result:

\begin{corollary}
There is no $\frac{1}{n^{1-\epsilon}}$, for any $\epsilon > 0$, factor approximation algorithm for the problem of finding maximum 0-1 timed matching problem in temporal graphs when each edge is associated with multiple time intervals unless NP=ZPP.
\end{corollary}

\subsection{Approximation Algorithm for Finding Maximum 0-1 Timed Matching for Temporal Graphs}
In this section, we present an approximation algorithm for finding a maximum 0-1 timed matching for temporal graphs. We first show that there exists an approximation preserving reduction from the maximum 0-1 timed matching (MAX-0-1-TM-MULT) problem  to the maximum independent set (MAX-IS) problem. We first define the MAX-0-1-TM-MULT problem.

\begin{definition}
	\textbf{MAX-0-1-TM-MULT:} Given a temporal graph $\mathcal{G = (V, E)}$ with lifetime $\mathcal{T}$, where each edge in $\mathcal{E}$ is associated with multiple time intervals, find a maximum sized 0-1 timed matching $M$ for $\mathcal{G}$.
\end{definition}

\begin{theorem}
	\label{thm:approxAlgo}
	There exists an approximation preserving reduction from  the MAX-0-1-TM-MULT problem to the MAX-IS problem.
\end{theorem}

\begin{proof}
	Consider an instance $\langle \mathcal{G = (V, E)} \rangle$ of MAX-0-1-TM-MULT problem where $\mathcal{G}$ is a temporal graph with $|\mathcal{V}| = n$, $|\mathcal{E}| = m$ and lifetime $\mathcal{T}$. We construct an instance of the MAX-IS problem as follows.
	
	We construct the static graph $G = (V, E)$ from $\mathcal{G = (V, E)}$ in the following way:
	
	\begin{itemize}
		\item For each edge $e_{uv} \in \mathcal{E}$, add a vertex $\nu_{uv} \in V$. Thus,
		
		$V := \{ \nu_{uv} \,|\, \forall e_{uv} \in \mathcal{E}\}$
		\item Add an edge $(\nu_{uv}, \nu_{vw})$ between $\nu_{uv}$ and $\nu_{vw}$ if $e_{uv}$ and $e_{vw}$ are overlapping with each other. Thus,
		
		$E := \{(\nu_{uv}, \nu_{vw}) \,|\, e_{uv}, e_{vw} \in \mathcal{E}$ are overlapping with each other$\}$     
	\end{itemize} 
	
	Given an independent set $I$ for $G$, we can construct a 0-1 timed matching $M$ for $\mathcal{G}$ as $ M = \{e_{uv} \,|\, \nu_{uv} \in I\}$. We show that $M$ is a 0-1 timed matching for $\mathcal{G}$ and $|I| = |M|$. As $M$ includes an edge for each vertex in $I$, $|I| = |M|$. We prove that $M$ is a 0-1 timed matching for $\mathcal{G}$ by contradiction. Let there be two edges $e_{uv}, e_{vw} \in M$ such that $e_{uv}$ and $e_{vw}$ are overlapping with each other. This implies that, $\nu_{uv}, \nu_{vw} \in I$  and $(\nu_{uv}, \nu_{vw}) \in E$. Thus, $I$ is not an independent set for $G$. This is a contradiction. Hence, $M$ is a 0-1 timed matching for $\mathcal{G}$. 
\end{proof}

In Theorem \ref{thm:approxAlgo}, it can be noted that the number of vertices in the constructed static graph $G = (V, E)$ is $|\mathcal{E}| = m$ and degree of each vertex in $V$ is the overlapping number of the corresponding edge in $\mathcal{E}$. Thus, the average degree of a vertex in $G$ is the average overlapping number of an edge in $\mathcal{G}$.

In \cite{halldorsson97}, an approximation algorithm is proposed to find a maximum independent set for a given graph $G$ with approximation ratio $\frac{5}{2d^* + 3}$ where $d^*$ is the average degree of a vertex in $G$. Thus, applying that algorithm on the static graph constructed from a given temporal graph using the steps described in Theorem \ref{thm:approxAlgo}, we can get the following result.

\begin{theorem}
	\label{thm:approx}
 There is an approximation algorithm to find a maximum 0-1 timed matching for a given temporal graph $\mathcal{G}$ with approximation ratio $\frac{5}{2\mathcal{N}^* + 3}$, where $\mathcal{N}^*$ is the average overlapping number of an edge in $\mathcal{G}$. 
 \end{theorem} 

It can be observed that the overlapping number of any edge $e_{uv} \in \mathcal{E}$ is dependent on the degree of $u$ and $v$ in the underlying graph $\mathcal{G}_U$. Thus, the average overlapping number $\mathcal{N}^*$ of edges in a temporal graph $\mathcal{G = (V, E)}$ is also dependent on the degree of vertices in $\mathcal{G}_U$. This indicates that this algorithm  produces a 0-1 timed matching with good approximation ratio for sparse graphs. The algorithm produces a 0-1 timed matching with good approximation ratio even when there exists a small number of edges for which the overlapping number is high if the average overlapping number of the edge in $\mathcal{E}$ is small.

It is proved earlier in Theorem \ref{thm:bdg3} that the problem of finding a maximum 0-1 timed matching for a  bounded degree bipartite temporal graph when degree of each vertex is bounded by 3 and each edge is associated with a single time interval is NP-complete. From Theorem \ref{thm:approx}, we obtain the following result for the problem of finding a maximum 0-1 timed matching for a given bounded degree temporal graph.

\begin{corollary}
	\label{cor:bounded}
	There is an approximation algorithm to find a maximum 0-1 timed matching  with approximation ratio $\frac{5}{2B+3}$ for a given bounded degree temporal graph where degree of each vertex is bounded by $B$ and each edge is associated with multiple time intervals. 
	Thus, for bounded degree temporal graphs where degree of each vertex is bounded by a constant, this is a constant factor approximation algorithm.
\end{corollary}

\section{Conclusion}
\label{conclusion}
In this paper, we have defined {\em 0-1 timed matching} on temporal graphs, and investigated the problem of finding a maximum 0-1 timed matching for different types of temporal graphs. We have proved that this problem is NP-complete for a rooted temporal tree when each edge is associated with $2$ or more time intervals, and proposed a $O(n \log n)$ time dynamic programming based algorithm  for a rooted temporal tree with $n$ vertices when each edge is associated with a single time interval. It is also proved that this problem is NP-complete for a bounded degree bipartite temporal graph where degree of each vertex is bounded by $3$ or more such that each edge is associated with a single time interval. We have also proved that there is no $\frac{1}{n^{1-\epsilon}}$, for any $\epsilon > 0$, factor approximation algorithm for the problem of finding a maximum 0-1 timed matching even for a rooted temporal tree when each edge is associated with multiple time intervals unless NP = ZPP. Then, we have proposed an approximation algorithm to address the problem for a temporal graph when each edge is associated with multiple time intervals. The work can be extended to consider the problem on other classes of temporal graphs.

\section*{Acknowledgement}

We gratefully acknowledge the comments and suggestions of the anonymous reviewers which have helped us to improve both the content and the presentation of the paper significantly.

\bibliographystyle{splncs03}
\bibliography{ref}

\begin{thebibliography}{10}
\providecommand{\url}[1]{\texttt{#1}}
\providecommand{\urlprefix}{URL }

\bibitem{akrida20}
Akrida, E.C., Deligkas, A., Mertzios, G.B., Spirakis, P.G., Zamaraev, V.:
  Matching in stochastically evolving graphs. arXiv preprint arXiv:2005.08263
  (2020)

\bibitem{Amblard11}
Amblard, F., Casteigts, A., Flocchini, P., Quattrociocchi, W., Santoro, N.: On
  the temporal analysis of scientific network evolution. In: International
  Conference on Computational Aspects of Social Networks, {CASoN}. pp. 169--174
  (2011)

\bibitem{bampis18}
Bampis, E., Escoffier, B., Lampis, M., Paschos, V.T.: Multistage matchings. In:
  Scandinavian Symposium and Workshops on Algorithm Theory, {SWAT}. pp. 7--1
  (2018)

\bibitem{Baste19}
Baste, J., Bui-Xuan, B.M., Roux, A.: Temporal matching. Theoretical Computer
  Science  806,  184--196 (2020)

\bibitem{Berman99}
Berman, P., Fujito, T.: On approximation properties of the independent set
  problem for low degree graphs. Theory of Computing Systems  32(2),  115--132
  (1999)

\bibitem{bertsekas81}
Bertsekas, D.P.: A new algorithm for the assignment problem. Mathematical
  Programming  21(1),  152--171 (1981)

\bibitem{bhalgat12}
Bhalgat, A., Feldman, J., Mirrokni, V.: Online allocation of display ads with
  smooth delivery. In: {ACM} International Conference on Knowledge Discovery
  and Data Mining, {SIGKDD}. pp. 1213--1221 (2012)

\bibitem{bravo19}
Bravo-Hermsdorff, G., Felso, V., Ray, E., Gunderson, L.M., Helander, M.E.,
  Maria, J., Niv, Y.: Gender and collaboration patterns in a temporal
  scientific authorship network. Applied Network Science  4(1),  1--17 (2019)

\bibitem{Xuan03}
Bui{-}Xuan, B., Ferreira, A., Jarry, A.: Computing shortest, fastest, and
  foremost journeys in dynamic networks. International Journal of Foundations
  of Computer Science  14(2),  267--285 (2003)

\bibitem{Caragiannis19}
Caragiannis, I., Filos{-}Ratsikas, A., Kanellopoulos, P., Vaish, R.: Stable
  fractional matchings. In: {ACM} Conference on Economics and Computation,
  {EC}. pp. 21--39 (2019)

\bibitem{CasteigtsHMZ20}
Casteigts, A., Himmel, A., Molter, H., Zschoche, P.: Finding temporal paths
  under waiting time constraints. In: International Symposium on Algorithms and
  Computation, {ISAAC}. pp. 30:1--30:18 (2020)

\bibitem{casteigts19}
Casteigts, A., Peters, J.G., Schoeters, J.: Temporal cliques admit sparse
  spanners. In: International Colloquium on Automata, Languages, and
  Programming, {ICALP}. pp. 134:1--134:14 (2019)

\bibitem{Cheriyan97}
Cheriyan, J.: Randomized {\~{o}}(m({\(\vert\)}v{\(\vert\)})) algorithms for
  problems in matching theory. {SIAM} Journal on Computing  26(6),  1635--1669
  (1997)

\bibitem{chimani20}
Chimani, M., Troost, N., Wiedera, T.: Approximating multistage matching
  problems. arXiv preprint arXiv:2002.06887  (2020)

\bibitem{cormen09}
Cormen, T.H., Leiserson, C.E., Rivest, R.L., Stein, C.: Introduction to
  algorithms. MIT press (2009)

\bibitem{Edmonds65}
Edmonds, J.: Paths, trees, and flowers. Canadian Journal of Mathematics  17,
  449–467 (1965)

\bibitem{Even75}
Even, S., Kariv, O.: An o(n{\^{}}2.5) algorithm for maximum matching in general
  graphs. In: Symposium on Foundations of Computer Science, {FOCS}. pp.
  100--112 (1975)

\bibitem{EvenT75}
Even, S., Tarjan, R.E.: Network flow and testing graph connectivity. {SIAM}
  Journal on Computing  4(4),  507--518 (1975)

\bibitem{Feng17}
Feng, H., Zhang, J., Wang, J., Xu, Y.: Dynamic analysis of {VANET} using
  temporal reachability graph. In: International Conference on Communication
  Technology, {ICCT}. pp. 783--787 (2017)

\bibitem{Ferreira02}
Ferreira, A.: {{On models and algorithms for dynamic communication networks:
  The case for evolving graphs}}. In: $4^{e}$ rencontres francophones sur les
  Aspects Algorithmiques des Telecommunications (ALGOTEL). pp. 155--161 (2002)

\bibitem{Ferreira10}
Ferreira, A., Goldman, A., Monteiro, J.: Performance evaluation of routing
  protocols for {MANET}s with known connectivity patterns using evolving
  graphs. Wireless Networks  16(3),  627--640 (2010)

\bibitem{garcia02}
Garc{\'\i}a, A., Hernando, C., Hurtado, F., Noy, M., Tejel, J.: Packing trees
  into planar graphs. Journal of Graph Theory  40(3),  172--181 (2002)

\bibitem{gupta14}
Gupta, A., Talwar, K., Wieder, U.: Changing bases: Multistage optimization for
  matroids and matchings. In: International Colloquium on Automata, Languages,
  and Programming, {ICALP}. pp. 563--575 (2014)

\bibitem{halldorsson97}
Halld{\'o}rsson, M.M., Radhakrishnan, J.: Greed is good: Approximating
  independent sets in sparse and bounded-degree graphs. Algorithmica  18(1),
  145--163 (1997)

\bibitem{halpern74}
Halpern, J., Priess, I.: Shortest path with time constraints on movement and
  parking. Networks  4(3),  241--253 (1974)

\bibitem{Han04}
Han, J.D.J., Bertin, N., Hao, T., Goldberg, D.S., Berriz, G.F., Zhang, L.V.,
  Dupuy, D., Walhout, A.J., Cusick, M.E., Roth, F.P., Marc, V.: Evidence for
  dynamically organized modularity in the yeast protein--protein interaction
  network. Nature  430(6995),  88--93 (2004)

\bibitem{Ho12}
Ho, C., Vaughan, J.W.: Online task assignment in crowdsourcing markets. In:
  Conference on Artificial Intelligence, {AAAI}. pp. 45--51 (2012)

\bibitem{Hopcroft71}
Hopcroft, J.E., Karp, R.M.: A n{\^{}}5/2 algorithm for maximum matchings in
  bipartite graphs. In: Symposium on Switching and Automata Theory, {SWAT}. pp.
  122--125 (1971)

\bibitem{Huang15}
Huang, S., Fu, A.W., Liu, R.: Minimum spanning trees in temporal graphs. In:
  {ACM} {SIGMOD}, International Conference on Management of Data. pp. 419--430
  (2015)

\bibitem{Hastad96}
Håstad, J.: Clique is hard to approximate within $n^{(1-\epsilon)}$. In: Acta
  Mathematica. pp. 627--636 (1996)

\bibitem{iribarren09}
Iribarren, J.L., Moro, E.: Impact of human activity patterns on the dynamics of
  information diffusion. Physical Review Letters  103(3),  038702 (2009)

\bibitem{Kameda74}
Kameda, T., Munro, J.I.: A o({\(\vert\)}v{\(\vert\)}*{\(\vert\)}e{\(\vert\)})
  algorithm for maximum matching of graphs. Computing  12(1),  91--98 (1974)

\bibitem{kierstead09}
Kierstead, H.A., Kostochka, A.V.: Efficient graph packing via game colouring.
  Combinatorics, Probability \& Computing  18(5),  765 (2009)

\bibitem{kleinberg06}
Kleinberg, J., Tardos, E.: Algorithm design. Pearson Education India (2013)

\bibitem{Kostakos09}
Kostakos, V.: Temporal graphs. Physica A  388(6),  1007--1023 (2009)

\bibitem{Le14}
Le, V.B., Pfender, F.: Complexity results for rainbow matchings. Theoretical
  Computer Science  524,  27--33 (2014)

\bibitem{lebre10}
Lebre, S., Becq, J., Devaux, F., Stumpf, M.P., Lelandais, G.: Statistical
  inference of the time-varying structure of gene-regulation networks. BMC
  Systems Biology  4(1),  130 (2010)

\bibitem{loebl90}
Loebl, M., Poljak, S.: Subgraph packing—a survey. In: Topics in Combinatorics
  and Graph Theory, pp. 491--503. Springer (1990)

\bibitem{lordan20}
Lordan, O., Sallan, J.M.: Dynamic measures for transportation networks. Plos
  one  15(12),  e0242875 (2020)

\bibitem{Mandal18}
Mandal, S., Gupta, A.: Approximation algorithms for permanent dominating set
  problem on dynamic networks. In: International Conference on Distributed
  Computing and Internet Technology, {ICDCIT}. pp. 265--279 (2018)

\bibitem{Mandal20}
Mandal, S., Gupta, A.: 0-1 timed matching in bipartite temporal graphs. In:
  International Conference on Algorithms and Discrete Applied Mathematics,
  {CALDAM}. pp. 331--346 (2020)

\bibitem{MandalG20cctree}
Mandal, S., Gupta, A.: Convergecast tree on temporal graphs. International
  Journal of Foundations of Computer Science  31(3),  385--409 (2020)

\bibitem{mertzios20}
Mertzios, G.B., Molter, H., Niedermeier, R., Zamaraev, V., Zschoche, P.:
  Computing maximum matchings in temporal graphs. In: Symposium on Theoretical
  Aspects of Computer Science, {STACS}. pp. 27:1--27:14 (2020)

\bibitem{Micali80}
Micali, S., Vazirani, V.V.: An o(sqrt({\(\vert\)}v{\(\vert\)})
  {\(\vert\)}e{\(\vert\)}) algorithm for finding maximum matching in general
  graphs. In: Symposium on Foundations of Computer Science, {FOCS}. pp. 17--27
  (1980)

\bibitem{Michail16}
Michail, O., Spirakis, P.G.: Traveling salesman problems in temporal graphs.
  Theoretical Computer Science  634,  1--23 (2016)

\bibitem{moser09}
Moser, H.: A problem kernelization for graph packing. In: International
  Conference on Current Trends in Theory and Practice of Computer Science,
  {SOFSEM}. pp. 401--412 (2009)

\bibitem{Mucha04}
Mucha, M., Sankowski, P.: Maximum matchings via gaussian elimination. In:
  Symposium on Foundations of Computer Science, {FOCS}. pp. 248--255 (2004)

\bibitem{Mucha06}
Mucha, M., Sankowski, P.: Maximum matchings in planar graphs via gaussian
  elimination. Algorithmica  45(1),  3--20 (2006)

\bibitem{Mulmuley87}
Mulmuley, K., Vazirani, U.V., Vazirani, V.V.: Matching is as easy as matrix
  inversion. In: {ACM} Symposium on Theory of Computing, {STOC}. pp. 345--354
  (1987)

\bibitem{rao07}
Rao, A., Hero, A.O., States, D.J., Engel, J.D.: Inferring time-varying network
  topologies from gene expression data. EURASIP Journal on Bioinformatics and
  Systems Biology  2007,  1--12 (2007)

\bibitem{spivey04}
Spivey, M.Z., Powell, W.B.: The dynamic assignment problem. Transportation
  Science  38(4),  399--419 (2004)

\bibitem{Wu14}
Wu, H., Cheng, J., Huang, S., Ke, Y., Lu, Y., Xu, Y.: Path problems in temporal
  graphs. Proceedings of the VLDB Endowment  7(9),  721--732 (2014)

\bibitem{xie16}
Xie, D., Wang, X., Liu, L., Ma, L.: Exploiting time-varying graphs for data
  forwarding in mobile social delay-tolerant networks. In: IEEE/ACM
  International Symposium on Quality of Service, {IWQoS}. pp. 1--10 (2016)

\bibitem{zhang18}
Zhang, W., Wei, C., Meng, X., Hu, Y., Wang, H.: The whole-page optimization via
  dynamic ad allocation. In: Companion of the The Web Conference, {WWW}. pp.
  1407--1411 (2018)

\bibitem{zschoche20}
Zschoche, P., Fluschnik, T., Molter, H., Niedermeier, R.: The complexity of
  finding small separators in temporal graphs. Journal of Computer and System
  Sciences  107,  72--92 (2020)

\end{thebibliography}

\end{document}